\documentclass[lettersize,journal]{IEEEtran}

\usepackage{amsmath,amssymb}

\newcommand{\R}{\mathbb{R}}

\renewcommand{\Re}[1]{\operatorname{Re}\left\{#1\right\}}

\newcommand{\vct}[1]{\boldsymbol{#1}}
\newcommand{\mtx}[1]{\mathbf{#1}}

\newcommand{\Ind}[1]{\mathds{1}\left[#1\right]}

\newcommand{\vp}{\vct{p}}
\newcommand{\vq}{\vct{q}}

\newcommand{\vv}{\vct{v}}

\newcommand{\vx}{\vct{x}}

\newcommand{\valpha}{\vct{\alpha}}

\newcommand{\vtheta}{\vct{\theta}}

\newcommand{\mK}{\mtx{K}}

\newcommand{\mS}{\mtx{S}}

\newcommand{\mId}{\mathbf{I}}

\newcommand{\vxbar}{\underline{\vct{x}}}

\DeclareMathOperator*{\argmin}{arg\,min}

\newcommand{\norm}[1]{\left\lVert#1\right\rVert}
\usepackage{amsfonts}
\usepackage{ dsfont }
\usepackage{wrapfig}
\usepackage{balance}

\usepackage{graphicx}
\usepackage[export]{adjustbox} 

\newcommand{\dpdth}{\frac{\partial \boldsymbol{p}}{\partial \boldsymbol{\theta}}}
\newcommand{\dqdv}{\frac{\partial \boldsymbol{q}}{\partial \boldsymbol{v}}}
\newcommand{\dqdth}{\frac{\partial \boldsymbol{q}}{\partial \boldsymbol{\theta}}}
\newcommand{\dpdv}{\frac{\partial \boldsymbol{p}}{\partial \boldsymbol{v}}}
\newcommand{\dvdp}{\frac{\partial \boldsymbol{v}}{\partial \boldsymbol{p}}}
\newcommand{\dvdq}{\frac{\partial \boldsymbol{v}}{\partial \boldsymbol{q}}}

\usepackage{amsmath,amsfonts}
\usepackage{algorithmic}
\usepackage{array}
\usepackage[caption=false,font=normalsize,labelfont=sf,textfont=sf]{subfig}
\usepackage{textcomp}
\usepackage{stfloats}
\usepackage{url}
\usepackage{verbatim}
\usepackage{graphicx}

\def\BibTeX{{\rm B\kern-.05em{\sc i\kern-.025em b}\kern-.08em
    T\kern-.1667em\lower.7ex\hbox{E}\kern-.125emX}}
\usepackage{balance}

\usepackage{amsthm}

\usepackage[dvipsnames]{xcolor}
\usepackage{multirow}
\usepackage{afterpage,refcount}

\newtheorem{theorem}{Theorem}

\newtheorem{lemma}{Lemma}
\newtheorem{remark}{Remark}
\newtheorem{assumption}{Assumption}

\usepackage{steinmetz}

\usepackage{hyperref}
\hypersetup{
    colorlinks=false,
    linkcolor=blue,
    filecolor=magenta,      
    urlcolor=cyan,
    pdftitle={tpwrs-sensitivities},
    pdfpagemode=FullScreen,
    }

\newcommand{\stildestar}{\mathbf{S}_{\ddagger}}
\newcommand{\stildedag}{\mathbf{S}_{\dagger}}

\usepackage{cite}

\begin{document}
\title{Conditions for Estimation of Sensitivities of\\Voltage Magnitudes to Complex Power Injections}
\author{Samuel Talkington$^{\dagger *}$,~\IEEEmembership{Student Member,~IEEE,}
Daniel Turizo$^{\ddagger}$,~\IEEEmembership{Member,~IEEE,}\\
Santiago Grijalva$^{*}$,~\IEEEmembership{Senior Member,~IEEE,}
Jorge Fernandez$^{\S}$,~\IEEEmembership{Student Member,~IEEE,}\\
and Daniel K. Molzahn$^{\ddagger}$,~\IEEEmembership{Senior Member,~IEEE.}
\thanks{The authors are with the School of Electrical and Computer Engineering, Georgia Institute of Technology, Atlanta, GA, USA.\\
Email:  \{talkington,djturizo,sgrijalva6,jlf,molzahn\}@gatech.edu.  

Manuscript submitted May 5, 2022; revised August 24th, 2022 and November 11th, 2022; accepted December 30th, 2022. This material is based upon work supported in part by the following organizations:}

\thanks{$^{*}$U.S. Department of Energy’s Office of Energy Efficiency and Renewable Energy (EERE) under Solar Energy Technologies Office (SETO) Agreement Number AWD-38426. Sandia National Laboratories is a multi-mission laboratory managed and operated by National Technology and Engineering Solutions of Sandia, LLC., a wholly owned subsidiary of Honeywell International, Inc., for the U.S. Department of Energy’s National Nuclear Security Administration under contract DENA0003525. 

$^{\dagger}$National Science Foundation Graduate Research Fellowship Program under Grant No. DGE-1650044. Any opinions, findings, and conclusions or recommendations expressed in this material are those of the author(s) and do not necessarily reflect the views of the National Science Foundation.}

\thanks{$^{\ddagger}$National Science Foundation (NSF) Energy, Power, Control and Networks program under award 202314 to the Georgia Tech Research Corporation. Any opinions, findings, and conclusions or recommendations expressed in this material are those of the author(s) and do not necessarily reflect the views of the National Science Foundation.}

\thanks{$^{\S}$Georgia Tech Research Corporation, Strategic Energy Institute, Energy Policy and Innovation (EPI) Center, Award DE00017791. 
This paper describes objective technical results and analysis. \\
Any subjective views or opinions that might be expressed in the paper do not necessarily represent the views of the sponsoring organizations.}

}

\maketitle 
\begin{abstract}
    Voltage phase angle measurements are often unavailable from sensors in distribution networks and transmission network boundaries. Therefore, this  paper  addresses  the  conditions for  estimating sensitivities  of  voltage  magnitudes  with  respect  to  complex (active  and  reactive) electric power injections based on sensor measurements.  These  sensitivities  represent  submatrices  of  the inverse power flow Jacobian. We extend previous results to show that  the  sensitivities  of  a  bus  voltage  magnitude  with  respect  to active power injections are unique and different from those with respect  to  reactive  power. The classical Newton-Raphson power flow model is used to derive a novel representation of bus voltage magnitudes as an underdetermined linear operator of the active and reactive power injections\textemdash parameterized by the bus power factors. Two conditions that ensure the existence of unique complex power injections given voltage magnitudes are established for this underdetermined linear system, thereby compressing the solution space. The first is a sufficient condition based on the bus power factors. The second is a necessary and sufficient condition based on the system eigenvalues.  We use matrix completion theory to develop estimation methods for recovering sensitivity matrices with  varying levels  of  sensor  availability.  Simulations  verify  the  results  and demonstrate engineering use of the proposed methods.
\end{abstract}
\begin{IEEEkeywords}
Power flow Jacobian, voltage sensitivity matrix, reactive power, underdetermined systems, inverse problems
\end{IEEEkeywords}

\newpage
\section{Introduction}
\IEEEPARstart{T}{he} power flow Jacobian matrix is central to many optimization, security, operation, and planning applications in electric power systems \cite{kundur}. This matrix contains the partial derivatives of the active and reactive AC power flow mismatch equations with respect to voltage magnitudes and phase angles, and is typically very sparse \cite{crow}. Knowledge of this matrix allows engineers to model the impact of changes in active and reactive power injections on the state of the system, and is central to the Newton-Raphson method to iteratively solve the non-linear AC power flow problem \cite{grainger_power_1994}.

The entries of the power flow Jacobian matrix describe the change in the active and reactive power injections, $\Delta \boldsymbol{p}$ and $\Delta \boldsymbol{q}$, due to small changes in the voltage magnitude $\Delta \boldsymbol{v}$ or angle $\Delta \boldsymbol{\theta}$, often known as sensitivity coefficients \cite{chen_measurement-based_2016,chen_measurement-based_2014,peschon_sensitivity_1968,chang_mutual_dependence}. There is a broad literature on the computation and applications of Jacobian matrices \cite{baker_network-cognizant_2018} or sensitivity coefficients in domains within power systems and others \cite{mugnier_model-lessmeasurement-based_2016,da_silva_data-driven_2020,nowak_measurement-based_2020,moffat_local_2021,hu_generalized_2019,sauer_reconstructing_1999,barter_closed_2021,sauer_jacobain}. 

In many cases, such as multi-area transmission systems, transmission system boundaries, and distribution systems, the network models needed to compute this matrix may not be complete or accurate \cite{inverse_power_flow}, especially in distribution systems \cite{deka_structure_learning}. Growing sensor deployment in electric power systems has spurred research on methods to recover sensitivity coefficients from measurements, allowing for the network behavior to be approximated even when the model is inaccurate, out of date, or unavailable \cite{mugnier_model-lessmeasurement-based_2016,chen_measurement-based_2014,chen_measurement-based_2016}. The literature has also explored estimation of the admittance matrix, both with synchrophasor measurements \cite{gupta_compound_2021,inverse_power_flow} and without \cite{zhang_phaseless_topology_line_param,claeys_unbalanced_phaseless_line_param}.



The net power injections at buses $i=1,\dots,n$ of an electric power system are denoted by $p_i + j q_i \in \mathbb{C}$, where $p_i$ is the net active power injection, $q_i$ is the net reactive power injection, and $j\triangleq \sqrt{-1}$. These are related with the bus voltages $\bar{v}_i \triangleq v_i \phase{\theta_i} \in \mathbb{C}$ and the net current injections $\bar{\ell}_i = \ell_i \phase{\phi_i} \in \mathbb{C}$ as
\begin{equation}
\label{eq:complex_power}
    p_i+jq_i = \bar{v}_i \bar{\ell}_i^{*} = \sqrt{p_i^2 + q_i^2} \phase{\theta_i - \phi_i}, 
\end{equation}
where $\bar{\ell}_i^*$ is the complex conjugate of the net current injection, $\theta_i-\phi_i$ is the difference between the phase angles of the voltage and current at bus $i$, and $\sqrt{p_i^2 + q_i^2}$ is the apparent power, i.e., the magnitude of the complex powers.

However, measurements of $\theta_i$ are often unavailable. For example, advanced metering infrastructure (AMI) data
are usually only available for the voltage magnitudes $\boldsymbol{v} \in \mathbb{R}^n$ and the active and reactive power injections $\boldsymbol{p},\boldsymbol{q} \in \mathbb{R}^n$ \cite{AMI_state_estimation}. This makes it difficult to model \eqref{eq:complex_power} with these measurement data, as they lack the voltage phase angles. 

Motivated by this problem, we are interested in determining the relationships between $\boldsymbol{v}$, $\boldsymbol{p}$, and $\boldsymbol{q}$ when $\boldsymbol{\theta}$ is not available. We propose that the \emph{power factors} at buses $i=1,\dots,n$, which we define as  
\begin{equation}
\label{eq:pf_def}
    \alpha_i \triangleq \cos(\theta_i - \phi_i) = \frac{p_i}{\sqrt{p_i^2+q_i^2}} = \cos\arctan \frac{q_i}{p_i},
\end{equation}
help provide an answer to this problem. These quantities are the ratios of active powers $p_i$ to the apparent powers $\sqrt{p_i^2 + q_i^2}$, which can be efficiently computed from $\boldsymbol{p},\boldsymbol{q}$ \cite{talkington_recovering_2021}.


Specifically, in this paper, we use \eqref{eq:pf_def} to determine a sufficient condition for when it is possible to relate the changes in $\boldsymbol{v}$ to changes in $\boldsymbol{p}$ and $\boldsymbol{q}$ via an \emph{underdetermined} system of equations that arises from the Newton-Raphson power flow model. In particular, when the injection power factors are known or can be found, it is possible to effectively solve this underdetermined system without the phase angles. Our work contributes to the existing literature that studies this ``voltage angle free" data input assumption, which has developed line and topology parameter estimation for single-phase\cite{zhang_phaseless_topology_line_param} and multi-phase unbalanced \cite{claeys_unbalanced_phaseless_line_param} networks. 

In addition to these contributions, we also develop specialized applications of matrix recovery algorithms for improving knowledge about the behavior of low-observability power systems in terms of the sensitivities of voltage magnitudes to active and reactive power injections. The key idea of these algorithms is that matrices with skewed spectral content, i.e., rapidly decreasing singular values, can be recovered in settings that are intractable with traditional estimation algorithms \cite{candes_exact_2009,yao_nuclear_norm,davenport_overview_2016,candes_matrix_completion_with_noise}. Recent research has shown the effectiveness of this class of algorithms in power system estimation problems, because many commonly encountered data matrices in electric power systems have rapidly decreasing singular values. Example applications include estimating voltage phasors \cite{donti_matrix_2020} and evaluating voltage stability \cite{lim_svd-based_2016}. We propose that the voltage magnitude blocks of the inverse power flow Jacobian also meet this criteria. Our work relates to previous work on adaptive power flow linearizations \cite{misra_optimal_2018} and a broad literature on measurement-based estimation of sensitivity coefficients \cite{chen_measurement-based_2014,chen_measurement-based_2016,kumar_neumann_voltage_2022}, their use in control \cite{gupta_model-less_2022}, and explorations into low-rank and online variants of these algorithms \cite{ospina_sensitivity_estimation}. 

In summary, the contributions of this paper are:
\begin{enumerate}
    \item Extending the results in \cite{christakou_efficient_2013}, conditions for radial distribution networks which ensure that the voltage magnitude sensitivities with respect to active and reactive power injections are unique. 
    \item Algorithms for radial distribution networks to recover or update the sensitivity matrices of voltage magnitudes to active and reactive power injections\textemdash which are submatrices of the inverse of the power flow Jacobian\textemdash via regression and matrix completion. The matrix recovery algorithms exploit the skewed spectral content of the inverse of the power flow Jacobian for primary networks.
    \item A sufficient condition for arbitrary networks that, if satisfied, guarantees the existence of a unique complex power injection state estimate from measurements of voltage magnitudes. This condition depends on the bus power factors $\boldsymbol{\alpha}$ and blocks of the power flow Jacobian.
   
\end{enumerate}

This paper assumes an unbalanced electrical network where the sets $\mathcal{S}$ and $\mathcal{N}$ contain the slack and PQ buses, respectively.
We also assume that voltage regulating devices are held fixed throughout the system. The analytical results in Sections \ref{sec:review_vph_sens} and \ref{sec:unique_vmag_sens} apply to unbalanced radial distribution networks, and to arbitrary networks in Sections \ref{sec:encoding_reactive_power} and \ref{sec:power_factor_bound}. Numerical experiments for the analytical results are shown in Section \ref{sec:comp_test_power_factor_bound} for meshed transmission networks and radial distribution networks. Implementation of the matrix recovery algorithms are shown for radial distribution networks in Section \ref{sec:sens_matrix_recovery}.  

\section{Preliminaries}
\subsection{Data Input Assumptions}
As the primary application of this paper is in distribution systems, we will work with datasets $\mathcal{D}_i$ for PQ buses $i=1,\dots,n \triangleq |\mathcal{N}|$, of the form:
\begin{equation}
\label{ami_signal}
    \mathcal{D}_i \triangleq 
    \big\{(v_{i,t},p_{i,t},q_{i,t})\big\}_{t=1}^m,
\end{equation}
where $v_{i,t}, p_{i,t},$ and $q_{i,t}$ are the nodal voltage magnitude, net active, and net reactive power injection measurements, respectively, at bus~$i$ for time steps $t=1,\dots,m$. We will assume the errors of these sensors to be normally distributed, with variance that is on the order of 0.5\%. AMI sensors typically have errors between 0.07\% and 4\% depending on the power quality of the load \cite{nist_ami}. In the next section, we will drop the subscript $t$.


\subsection{The Newton-Raphson Power Flow}
Consider the power balance equations for a bus $i \in \mathcal{N}$:
\begin{align}
    \label{AC_mismatch}
        p_i &= v_i \sum_{k=1}^{n}  v_k \big(G_{ik} \cos \left( \theta_i- \theta_k \right) + B_{ik}\sin (\theta_i - \theta_k)\big), \\
    \label{AC_mismatch_2}
        q_i &= v_i \sum_{k=1}^{n}  v_k \big(G_{ik} \sin (\theta_i - \theta_k) - B_{ik}\cos \left( \theta_i- \theta_k \right)\big),
\end{align}
where $v_i,v_k$ are the voltage magnitudes at buses $i$ and $k$ and $G_{ik},B_{ik}$ are the real and imaginary parts of the $ik$-th entry of the bus admittance matrix, $Y_{ik} = G_{ik} + j B_{ik}$. In order to solve the systems \eqref{AC_mismatch} and \eqref{AC_mismatch_2}, a classical approach is the Newton-Raphson (NR) algorithm, which iteratively solves the system of equations \eqref{eq:classical_NR_iteration}:
\begin{equation}\arraycolsep=1.4pt\def\arraystretch{1.75}
\label{eq:classical_NR_iteration}
    \underbrace{
    \begin{bmatrix}
        \Delta \boldsymbol{p}\\
        \Delta \boldsymbol{q}
    \end{bmatrix}
    }_{(2n \times 1)} = 
    \underbrace{\left[
        \begin{array}{c|c}
        \dpdth & \dpdv   \\
        \hline 
        \dqdth & \dqdv \\
        \end{array}
    \right]}_{(2n \times 2n)}
      \underbrace{
        \begin{bmatrix}
        \Delta \boldsymbol{\theta}\\
        \Delta \boldsymbol{v}
        \end{bmatrix}}_{(2n \times 1)} 
      = \mathbf{J}
        \begin{bmatrix}
        \Delta \boldsymbol{\theta}\\
        \Delta \boldsymbol{v}
        \end{bmatrix},
\end{equation}
where $\Delta \boldsymbol{p},\Delta \boldsymbol{q} \in \mathbb{R}^{n}$ are vectors of small deviations in active and reactive power, respectively. The power flow Jacobian $\mathbf{J}$ is known to be relatively constant with respect to small changes in power injections \cite{chen_measurement-based_2014,chen_measurement-based_2016}.
Consider the block submatrices of the inverse power flow Jacobian. Hereafter, we refer to blocks of the Jacobian as sensitivity matrices and their elements as sensitivity coefficients. Denote the blocks of the inverse as $\mathbf{S}^{x}_y \in \mathbb{R}^{n \times n}$.  The inverse problem of \eqref{eq:classical_NR_iteration} can be written as:
\begin{equation}\arraycolsep=1.4pt\def\arraystretch{1.75}
\label{nr_powerflow}
    \underbrace{
        \begin{bmatrix}
        \Delta \boldsymbol{\theta} \\
        \Delta \boldsymbol{v} 
        \end{bmatrix}}_{(2n \times 1)}
    = \underbrace{\left[
        \begin{array}{c|c}
        \mathbf{S}^{\theta}_{p}  & \mathbf{S}^{\theta}_q   \\
        \hline 
        \mathbf{S}^{v}_p  & \mathbf{S}^{v}_q 
        \end{array}
    \right]}_{(2n \times 2n)}
     \underbrace{
    \begin{bmatrix}
        \Delta \boldsymbol{p} \\
        \Delta \boldsymbol{q} 
    \end{bmatrix}
    }_{ (2n \times 1)}= \mathbf{J}^{-1}
    \begin{bmatrix}
        \Delta \boldsymbol{p}\\
        \Delta \boldsymbol{q}
    \end{bmatrix}.
\end{equation}


\subsection{Phaseless Approximation of the Power Flow Equations}
Sensing devices that provide phase angle measurements have well-known benefits and applications. However, the large-scale deployment and application of such measurements continues to be heterogeneous and challenging due to infrastructure costs and communication requirements. 

Particularly in distribution systems, access to phase angle information $\Delta \boldsymbol{\theta}$ may be unavailable due to low penetrations of phasor measurement units (PMUs) \cite{naspi_pmus_2021_tech_report}, making it impossible to realistically solve the system of equations \eqref{eq:classical_NR_iteration} and \eqref{nr_powerflow}. Furthermore, the system of equations needs to be solved in real time. The sensitivity matrices are relatively constant inter-temporally, allowing for model behavior to be linearly approximated \cite{chen_measurement-based_2014,chen_measurement-based_2016}. The voltage magnitude of bus $i$, $v_i$, can be written as a first-order linear approximation around a given operating condition:
\begin{equation}
    \label{vectorized_linear_sensitivity_model}
    v_i \approx v_i^0 + \frac{\partial v_i}{\partial \boldsymbol{p} } \Delta \boldsymbol{p} + \frac{\partial v_i}{\partial \boldsymbol{q} } \Delta  \boldsymbol{q},
\end{equation}
where $v_i^0$ is the voltage magnitude of bus $i$ in the given operating condition and 
$\frac{\partial v_i}{\partial \boldsymbol{p}},\frac{\partial v_i}{\partial  \boldsymbol{q}} \in \mathbb{R}^{1 \times n}$
are the $i$-th rows of the matrices describing voltage magnitude sensitivities with respect to the power injections. From \eqref{nr_powerflow}, we can write a rectangular linearized system which relates voltage magnitude variations to active and reactive power variations: 
\begin{equation}\arraycolsep=1.4pt\def\arraystretch{1.5}
\label{eq:underdetermined_nr_inv_system}
    \underbrace{\Delta\boldsymbol{v} }_{(n \times 1)} = \underbrace{
    \begin{bmatrix}
         \mathbf{S}^{v}_p  & \mathbf{S}^{v}_q 
    \end{bmatrix}
    }_{(n\times 2n)}
    \underbrace{
    \begin{bmatrix}
       \Delta \boldsymbol{p} \\
        \Delta \boldsymbol{q} 
    \end{bmatrix}
    }_{(2n \times 1)} = 
    \underbrace{\tilde{\mathbf{S}} \Delta \boldsymbol{x}}_{(n \times 1)} .
\end{equation}
The matrix $\tilde{\mathbf{S}}$ describes the sensitivities of the $n$ voltage magnitudes to $n$ active and $n$ reactive power injections and is the main quantity of interest in this paper. Note that $\tilde{\mathbf{S}} \in \mathbb{R}^{n \times 2n}$, thus, $\operatorname{rank}(\tilde{\mathbf{S}}) \leq n,$ and $\operatorname{rank}(\tilde{\mathbf{S}}) + \dim(\operatorname{null}(\tilde{\mathbf{S}})) = n$, which implies that $\dim(\operatorname{null}(\tilde{\mathbf{S}})) >0$. Thus, in general, there are \textit{infinitely many solutions} $\Delta \boldsymbol{x}$ to the system of equations \eqref{eq:underdetermined_nr_inv_system}. Throughout the next section, we will show how and when we can circumvent this mathematical problem by exploiting knowledge of power system physics.

\section{Analysis of Voltage Sensitivities}
\label{sec:analytical_results}
\newcommand{\Vph}{\bar{v}}
\newcommand{\Sph}{\tilde{\mathbf{S}}}
\newcommand{\Ybus}{\boldsymbol{Y}}
Prior numerical results have empirically indicated that the voltage magnitude sensitivities for active and reactive power injections are distinct \cite{lin_data-driven_2021,talkington_power_2021}. Very recently, under the assumption that $\theta_i =0$, \cite{chang_mutual_dependence} showed that $\frac{\partial v_i}{\partial p_i}$ is correlated to $\frac{\partial v_i}{\partial q_i}$. In this section, we extend \cite{christakou_efficient_2013} to show that the voltage magnitude sensitivities to active and reactive power injections are uncorrelated in general in unbalanced networks.

\subsection{Review of Voltage Phasor Sensitivities}
\label{sec:review_vph_sens}
The net complex power injection at bus~$i$,~$p_i + j q_i$, is related to the network's phasor voltages via the power flow equations:
\begin{equation}
    \label{eq:power_flow_equations}
  p_i + j q_i = \bar{v}_i \Bigg( \sum_{k \in \mathcal{N} \cup \mathcal{S}} Y_{ik} \bar{v}_k \Bigg)^* \quad \forall  i \in \mathcal{N},
\end{equation}
where $\bar{v}_i$ is the voltage phasor of bus $i$ and $(\cdot)^*$ denotes the complex conjugate. Following \cite{christakou_efficient_2013}, differentiating \eqref{eq:power_flow_equations} with respect to active and reactive power individually yields the systems of equations \eqref{eq:vph_sensitivities_definition}, whose solutions are the sensitivities of \textit{phasor} voltages to active and reactive power injections:
\begin{subequations}
\label{eq:vph_sensitivities_definition}
\begin{align}
    \Ind{i=l} =\frac{\partial \bar{v}_i^*}{\partial p_l} \sum_{k \in \mathcal{S} \cup \mathcal{N}} Y_{ik}  \bar{v}_k + \bar{v}_i^* \sum_{k \in \mathcal{N}}Y_{ik} \frac{\partial  \bar{v}_k}{\partial p_l}, \\
    \label{complex_sensitivities_2}
    -j\Ind{i=l}=\frac{\partial \bar{v}_i^*}{\partial q_l} \sum_{k \in \mathcal{S} \cup \mathcal{N}} Y_{ik} \bar{v}_k + \bar{v}_i^* \sum_{k \in \mathcal{N}}Y_{ik} \frac{\partial  \bar{v}_k}{\partial q_l},
\end{align}
\end{subequations}
where $\Ind{i=l}$ is the indicator function, defined as
\begin{equation}
    \Ind{i=l} \triangleq \left\{ {\begin{array}{*{20}l}
   {1} & {\text{if} \; \; i = l},  \\
   {0} & \text{otherwise}.  \\
\end{array}} \right.
\end{equation}
The voltage phasor sensitivities to active and reactive power injections, $\frac{\partial \bar{v}_i}{\partial p_l}$ and $\frac{\partial\bar{v}_i}{\partial q_l}$, are of particular interest in distribution systems since they have a unique solution.

\begin{remark}
\label{remark:unique_complex_vph_sens}
\cite{christakou_efficient_2013} In a radial distribution network, the nontrivial solutions of the equations in the systems of \eqref{eq:vph_sensitivities_definition}, i.e., where $\Ind{i=l} \neq 0$, the unknowns $\frac{\partial\bar{v}_i}{\partial q_l}$ and $\frac{\partial \bar{v}_i}{\partial p_l}$ achieve distinct complex values.
\end{remark}

For the next lemmas, we use $\operatorname{Re} \{\cdot\}$ and $\operatorname{Im} \{\cdot\}$ to represent the real and imaginary part of a complex number, respectively.

\begin{lemma}
\label{lemma:representation}
The voltage magnitude sensitivity coefficients of a network can be written as \eqref{magnitude_sens} and \eqref{mag_sens_2}. 
\begin{align}
\label{magnitude_sens}
   \frac{\partial v_i }{\partial q_l} &= \frac{1}{v_i} \operatorname{Re} \left\{\bar{v}_i^*\frac{\partial\bar{v}_i}{\partial q_l}\right\}, \\
    \label{mag_sens_2}
   \frac{\partial v_i }{\partial p_l} &= \frac{1}{v_i} \operatorname{Re} \left\{\bar{v}_i^*\frac{\partial \bar{v}_i}{\partial p_l}\right\}.
\end{align}
\end{lemma}
\begin{proof}
See Appendix \ref{apdx:proof_of_vmag_representation}.
\end{proof}

\subsection{Unique Voltage Magnitude Sensitivities}
\label{sec:unique_vmag_sens}
Next, we will show that if $\frac{\partial\bar{v}_i}{\partial q_l}$ and $\frac{\partial \bar{v}_i}{\partial p_l}$ have unique solutions, then we can say the same for the voltage magnitudes.

\begin{lemma}
\label{lemma:distinct_vmag_sens}
Let the rectangular form of the complex sensitivities be  $\frac{\partial\bar{v}_i}{\partial q_l} = a+jb$ and $\frac{\partial \bar{v}_i}{\partial p_l}=c+jd$ respectively.~If
\begin{align}
\label{zero_con_1}
    (a,b) \not\in \{(a,b) : \operatorname{Re}\{\bar{v}_i\}a+\operatorname{Im}\{\bar{v}_i\}b = 0 \},\\
    \label{zero_con_2}
    (c,d) \not\in \{(c,d) : \operatorname{Re}\{\bar{v}_i\}c+\operatorname{Im}\{\bar{v}_i\}d = 0\},
\end{align}
then $ \frac{\partial v_i}{\partial q_l} \neq \frac{\partial v_i}{\partial p_l} \ \forall i, l$.
\end{lemma}
\begin{proof}
    See Appendix \ref{apdx:proof_of_distinct_vmag_sens}.
\end{proof}

Lemma \ref{lemma:distinct_vmag_sens} implies that the matrices $\mathbf{S}^v_p$, $\mathbf{S}^v_q$ in \eqref{nr_powerflow} are full rank for a radial distribution network, and furthermore, that the matrix $\tilde{\mathbf{S}}$ has full column rank for any subset of the columns whose cardinality is less than $\frac{n}{2}$. Essentially, the voltage magnitude sensitivities to active and reactive power injections will always yield a unique solution, and it is possible to quantify changes in both active and reactive power injections using only voltage magnitudes.


\begin{remark}
Consider a bus $l\in \mathcal{N}$ in a radial distribution network with unknown complex power injections. Given a vector of voltage magnitude perturbations $\Delta \boldsymbol{v} \in \mathbb{R}^n$ and a tall matrix of sensitivities of the network voltage magnitudes to the active and reactive power injections at bus $l$,
\begin{equation}
   \mathbf{S}_{\perp} \triangleq \begin{bmatrix}
       \frac{\partial v_1}{\partial p_l} & \dots & \frac{\partial v_n}{\partial p_l}\\
       \frac{\partial v_1}{\partial q_l} & \dots & \frac{\partial v_n}{\partial q_l} 
    \end{bmatrix}^T \in \mathbb{R}^{n \times 2},
\end{equation}
then there is a unique least squares solution for the rectangular complex power perturbation $\Delta \boldsymbol{x} = \left[ \Delta p_l, \Delta q_l \right]^T$ such that
\begin{equation}
\label{eq:lsq}
    \Delta \boldsymbol{x} = \big(\mathbf{S}_{\perp}^T \mathbf{S}_{\perp} \big)^{-1} \mathbf{S}_{\perp}^T \Delta \boldsymbol{v},
\end{equation}
because, by Lemma \ref{lemma:distinct_vmag_sens}, the system $\mathbf{S}_{\perp} \Delta \boldsymbol{x} = \boldsymbol{0}$
   has a solution if and only if $\Delta \boldsymbol{x} = \boldsymbol{0}$. Therefore, $\mathbf{S}_{\perp}$ has full rank and \eqref{eq:lsq} will always exist.
\end{remark}

\subsection{Relating Active and Reactive Power Perturbations}
\label{sec:encoding_reactive_power}
In this section, we describe how to use the bus power factors to encode the impact of reactive power injections on voltage magnitudes as an equivalent active power injection.
\begin{lemma}
\label{lemma:K_matrix}
    Let $\valpha \triangleq [ \alpha_1,\dots,\alpha_n]^T \in \R^n$ be the bus power factors. Let $\Delta \vp\triangleq\vp - \vp_0 \in \R^n$ and $\Delta \vq \triangleq \vq - \vq_0 \in \R^n$  be vectors of active and reactive power perturbations around an operating point, where $\Delta \vp,\Delta \vq \neq \vct{0}$. Then:
    \begin{equation}
        \alpha_i \in (0,1) \  i =1,\dots,n \iff \exists \mK(\valpha): \R^n \mapsto  \R^{n \times n},
    \end{equation}
    such that $k(\alpha_i) \triangleq K_{ii} \triangleq \pm \frac{1}{\alpha_i} \left(1- \alpha_i^2\right)^{\frac{1}{2}}$ for $i=1,\dots,n$, and where we define $\mK(\valpha) \triangleq \operatorname{diag}(k(\boldsymbol{\alpha}))$ so that 
                \begin{align}
                    \Delta \vq(\Delta \vp| \valpha) = \mathbf{K}(\valpha)\Delta \vp,\\
                    \Delta \vp(\Delta \vq | \valpha) = \mathbf{K}^{-1}(\valpha)\Delta \vq,\\
                    \mathbf{K}(\valpha) \mathbf{K}^{-1}(\valpha) = \mK^{-1}(\valpha) \mK(\valpha) = \mId,
                \end{align}
                where $\mId$ is the identity matrix.
                \end{lemma}
                \begin{proof}
                    If $\alpha_i \in (0,1) \ \forall i, \ $ then we can write the net nodal active power injection as \eqref{eq:net-q-representation-impl}:
                    \begin{subequations}
                        \label{eq:net-q-representation-impl}
                         \begin{align}
                            p_i \triangleq \Re{p_i + j q_i} \triangleq \alpha_i \sqrt{p_i^2+q_i^2},
                            \\
                            \iff p_i^2 = \alpha_i^2 \left(p_i^2 + q_i^2\right)
                            \iff p_i^2 \left(1-\alpha_i^2\right) = \alpha_i^2 q_i^2,
                            \\
                            \iff q_i = \pm \alpha_i^{-1} (1-\alpha_i^2)^{\frac{1}{2}} p_i.
                    \end{align}
                    \end{subequations}
                    For the only if condition, consider that if $\exists i \in \{1,\dots,n\}$ such that $\alpha_i \notin (0,1)$; this implies that there exists an  $\alpha_i =1$ such that $K_{ii} = 0$ or $\alpha_i =0$ and $K_{ii}^{-1}=0$.
                    
                \end{proof}
\begin{lemma}
\label{lemma:implicit_representation}
Let $\valpha \triangleq [ \alpha_1,\dots,\alpha_n]^T \in (0,1]^n$ be the bus power factors. The network voltage deviation vector $\Delta \boldsymbol{v} \triangleq \vv - \vv_0$ can be written as:
$\Delta \boldsymbol{v} = \stildedag(\valpha)\Delta \boldsymbol{p},$
 where we define 
\begin{subequations}
\label{eq:implicit_representation_def}
\begin{align}
     \stildedag(\valpha) : \R^n \mapsto \R^{n \times n} \triangleq \left(  \mathbf{S}_p^{v}+ \mathbf{S}_q^v \mathbf{K}(\valpha)\right),\\
     \mathbf{K}(\valpha): \R^n \mapsto \R^{n \times n} \triangleq \operatorname{diag}\big( k(\boldsymbol{\alpha})\big),
\end{align}
\end{subequations}
and where $\operatorname{diag}(\vx) \in \R^{n \times n}$ denotes a diagonal matrix constructed from a vector $ \vx \in \R^n$ such that $\operatorname{diag}(\vx)[i,j] = 0 \ \forall i \neq j$ and  $\operatorname{diag}(\vx)[i,i] \triangleq x_i \ \forall \ i=1\dots,n$.
\end{lemma}

\begin{proof}
    By Lemma \ref{lemma:K_matrix}, we can express the reactive power injection as a function of the active power injection parameterized by the bus power factor:
    \begin{equation}
        q_i(p_i|\alpha_i) = k(\alpha_i) p_i =  
       \pm \frac{p_i}{\alpha_i} \left( 1- \alpha_i^2\right)^{\frac{1}{2}}.
    \end{equation}
    Therefore, we can express the voltage deviation vector $\Delta \boldsymbol{v} \in \mathbb{R}^n$ as \eqref{eq:voltage-dev-representation-impl}:
    \begin{subequations}
        \label{eq:voltage-dev-representation-impl}   
        \begin{align}
            \Delta \boldsymbol{v} &= \mathbf{S}_p^{v} \Delta \boldsymbol{p} + \mathbf{S}_q^v \mathbf{K}(\valpha) \Delta \boldsymbol{p},\\
            \label{eq:introducing_stilde_dagger}
        &=     \left(  \mathbf{S}_p^{v}+ \mathbf{S}_q^v \mathbf{K}(\valpha)\right) \Delta \boldsymbol{p} \triangleq \stildedag\Delta \boldsymbol{p},
        \end{align}
    \end{subequations}
    which is what we wanted to show. 
\end{proof}
\begin{figure}
    \centering
    \includegraphics[width=\linewidth,keepaspectratio]{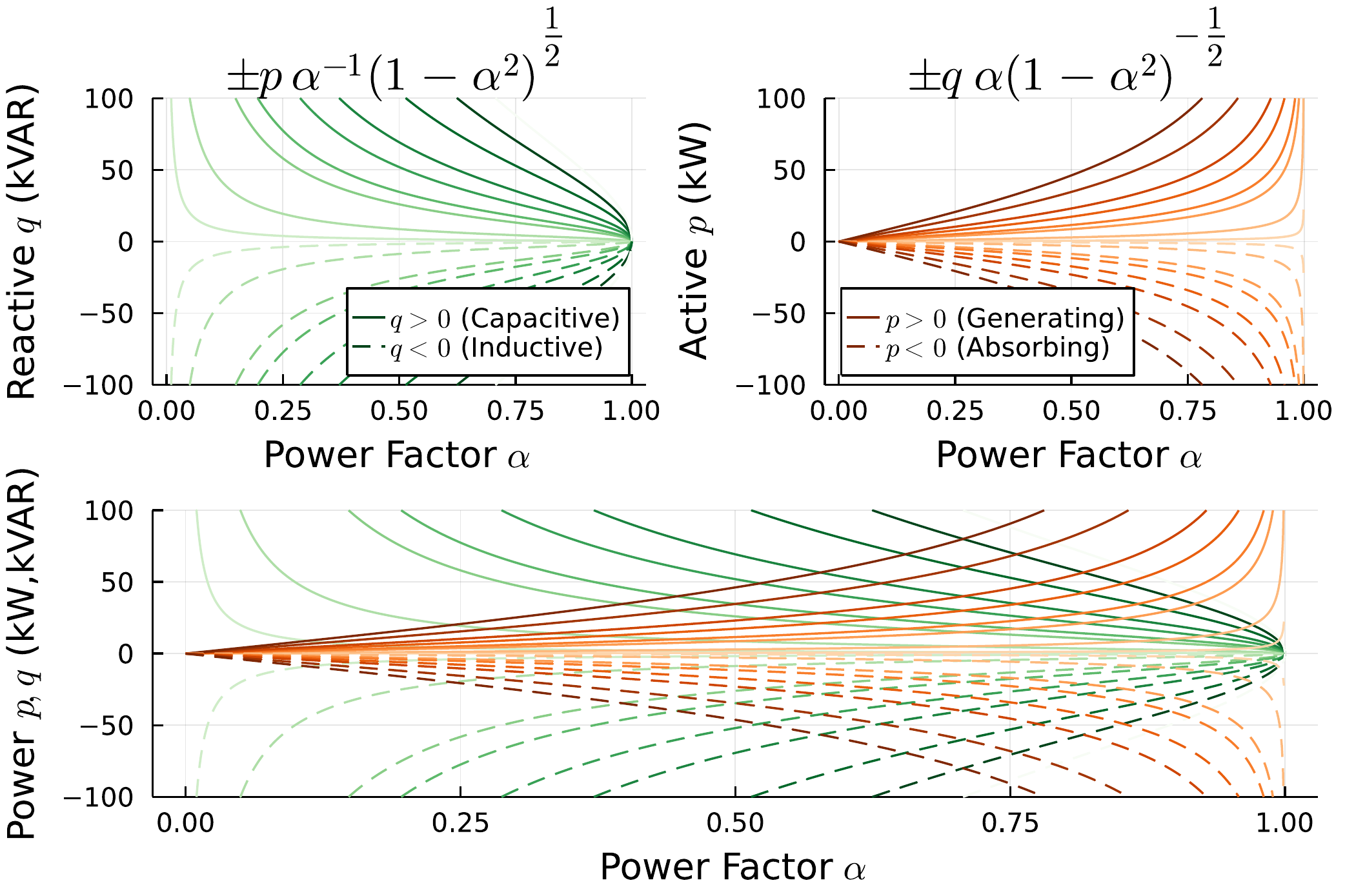}
    \caption{Representing the net reactive (green) and active (orange) power injections as parameterized functions of the power factor. Color intensity represents the size of the injection opposite that shown on the vertical axis.} 
    
    
    
    \label{fig:reactive_power_representation}
\end{figure}
This representation of reactive power is illustrated in Fig.~\ref{fig:reactive_power_representation}. Lemma \ref{lemma:implicit_representation} is useful because it allows us to represent the impacts of changes in reactive power on the voltage magnitudes as equivalent changes in active power by exploiting the relationship between active and reactive power for known bus power factors, which can be efficiently estimated from historical AMI data \cite{talkington_recovering_2021,talkington_power_2021}.

\section{Estimation conditions}
\label{sec:main_results}
    This section contains the main analytical results of this paper. Section \ref{sec:thm1_neumann} develops a sufficient condition for a unique solution to the underdetermined system $\Delta \vv = \tilde{\mS} \Delta \vx$. Section \ref{sec:power_factor_bound} develops bounds to analyze the feasible power factors that satisfy the sufficient condition in Section \ref{sec:thm1_neumann}. Section \ref{sec:thm2_eigval} develops a necessary and sufficient condition to solve the system $\Delta \vv = \tilde{\mS} \Delta \vx$ based on the singularity conditions of two linear operators derived from the bus power factors and the Newton-Raphson power flow Jacobian.
    \subsection{Neumann series-based sufficient conditions}
        \label{sec:thm1_neumann}
        Below, two Assumptions are stated that allow us to develop a sufficient condition for when the inverse of $\stildedag$ as defined in  \eqref{eq:introducing_stilde_dagger} is guaranteed to exist.
        \begin{assumption}
        \label{assum:nonsingular_jacobian}
        The power flow Jacobian is nonsingular and the unknown angle sensitivity submatrix $\dpdth$ is positive definite. 
        \end{assumption}
        The non-singularity of the Jacobian is a reasonable assumption, as power systems typically operate far from the point of voltage collapse. While counterexamples do exist (for example, as described in \cite{grijalva_singularity,hiskens_singularity,baker_jacobian_2013}), the full power flow Jacobian can typically be expected to be nonsingular if normal network operating conditions are assumed. The assumption of $\dpdth$ being positive definite is not restrictive, as it holds in most practical cases (see \cite{farivar_equilibrium_2013,zhu_fast_2016,admittance_invertibility}).

        \begin{assumption}
        \label{assum:k_difference_bound}
        The difference between the maximum and minimum elements of $\mathbf{K}(\valpha)$, 
        \begin{equation}
            \Delta k \triangleq k_{\rm max}-k_{\rm min} = \frac{\sqrt{1- \alpha_{\rm min}^2}}{\alpha_{\rm min}} - \frac{\sqrt{1-\alpha_{\rm max}^2}}{\alpha_{\rm max}},
        \end{equation}
        is sufficiently small relative to an expression that depends on the power-to-voltage-phase-angle sensitivity matrices $\dpdth$ and $\dqdth$, which will be defined explicitly in \eqref{eq:inv_condition_1} and \eqref{eq:inv_condition_2}.
        \end{assumption}
        
        

            \begin{theorem}
            \label{thm1:suff_cond}
                
                Let $\Delta \boldsymbol{v} \in \mathbb{R}^n$ be a vector of voltage magnitude perturbations. If Assumptions \ref{assum:nonsingular_jacobian} and \ref{assum:k_difference_bound} hold, there exists unique complex power perturbations $\Delta \boldsymbol{x} \triangleq [\Delta \boldsymbol{p}^T, \Delta \boldsymbol{q}^T]^T \in \mathbb{R}^{2n}$ in rectangular coordinates such that $\Delta \boldsymbol{v} = \tilde{\mathbf{S}} \Delta \boldsymbol{x}.$
            \end{theorem}
            \begin{proof}
               Using Lemma \ref{lemma:implicit_representation}, it now suffices to show that
            \begin{equation}
            \label{s_dagger}
                \stildedag \triangleq \left(  \mathbf{S}_p^{v}+ \mathbf{S}_q^v \mathbf{K}(\valpha) \right)
            \end{equation} 
            is invertible to complete the proof.
            
            If Assumption 1 holds, then both $\mathbf{J}$ and $\dpdth$ are invertible. Thus, we can apply the Schur Complement to write the reactive power voltage sensitivity matrix in terms of the blocks of the Jacobian in \eqref{nr_powerflow} as
            \begin{equation}
                \label{sqv_schur}
                \mathbf{S}_q^v = 
                \left( \dqdv - \dqdth \left({\dpdth}\right)^{-1} \dpdv \right)^{-1}, 
            \end{equation}
            and the active power voltage sensitivity matrix is then:
            {\small
            \begin{align}
                \mathbf{S}_p^v &= - 
                 \left( \dqdv - \dqdth \left({\dpdth}\right)^{-1} \dpdv \right)^{-1}\dqdth \left( \dpdth \right)^{-1},\\
                \label{spv_schur}
                &= - \mathbf{S}_q^{v} \dqdth \left( \dpdth \right)^{-1}.
            \end{align}
            }
            Combining \eqref{s_dagger}, \eqref{sqv_schur}, and \eqref{spv_schur}, we can express $\stildedag$ as: 
            \begin{align}
                \stildedag &= \mathbf{S}_q^v\left( \mathbf{K} - \dqdth \left( \dpdth\right)^{-1}\right).
            \end{align}
            Thus, we have that
            \begin{equation}
                (\mathbf{S}_q^v)^{-1} \stildedag \dpdth =  \mathbf{K} \dpdth -\dqdth.
            \end{equation}
            Let $k_{\rm max}$ and $k_{\rm min}$ denote the maximum and minimum entries of $\mathbf{K}$, respectively. Recall that we defined $\Delta k \triangleq k_{\rm max} - k_{\rm min}$ and let $\Delta \mathbf{K} \triangleq k_{\rm max} \mathbf{I} - \mathbf{K}$. Then we can write
            \begin{align}
                (\mathbf{S}_q^v)^{-1} \stildedag \dpdth &=  \underbrace{k_{\rm max} \dpdth -\dqdth}_{\triangleq \mathbf{M} \succ {\bf 0}} - \Delta \mathbf{K} \dpdth, \\
                (\mathbf{S}_q^v)^{-1} \stildedag \dpdth &= \mathbf{M} \left({\mathbf{I} - \mathbf{M}^{-1} \Delta \mathbf{K} \dpdth}\right). \label{eq:neumann}
            \end{align}
            The inverse of the term in parentheses in \eqref{eq:neumann} can be computed using Neumann series. According to \cite[Ch. 22, Lemma 1]{royden_real_2017}, this inverse is guaranteed to exist if:
            
        \begin{equation}
            \left\|{\mathbf{M}^{-1} \Delta \mathbf{K}\dpdth}\right\|_2 < 1,
        \end{equation}
        where $\left\|{\cdot}\right\|_2$ is the largest singular value\textemdash also known as the spectral norm or operator norm\textemdash of the argument. The sub-multiplicative property of this norm allows us to use the stronger inequality \eqref{eq:inv_condition_1}, where $\left\|{\Delta \mathbf{K}}\right\|_2 = k_{\rm max} - k_{\rm min} = \Delta k$:
        \begin{equation}
        \label{eq:inv_condition_1}
            \left\|{\mathbf{M}^{-1}}\right\|_2 \left\|{\Delta \mathbf{K}}\right\|_2 \left\|{\dpdth}\right\|_2 < 1.
        \end{equation}
        Therefore, the inverse is also guaranteed to exist if
        \begin{equation}
        \label{eq:inv_condition_2}
              \Delta k < \left\|{\mathbf{M}^{-1}}\right\|_2^{-1}  \left\|{\dpdth}\right\|_2^{-1},
        \end{equation}
        which holds for close enough power factors. In conclusion, the right hand side of \eqref{eq:neumann} is invertible, so the left hand side is invertible too. Under condition \eqref{eq:inv_condition_1} or  \eqref{eq:inv_condition_2}, both $(\mathbf{S}_q^v)^{-1}$ and $\dpdth$ are invertible. Then, $\stildedag$ is invertible too. This means that for any $\Delta \boldsymbol{v}$ we have a unique $\Delta \boldsymbol{p}$ and a unique $\Delta \boldsymbol{x}$.
            \end{proof}

    \subsection{Analyzing the power factor bound}
    \label{sec:power_factor_bound}
    \newcommand{\alphamin}{\alpha_{\rm min}}
    \newcommand{\alphamax}{\alpha_{\rm max}}
    \newcommand{\deltakmax}{\Delta k_{\text{max}}}
    For a given operating condition, if we have access to the full network model, the upper bound on $\Delta k$ can be computed directly. 
    To achieve this, recall that we defined the quantity $\Delta k$ as $\Delta k \triangleq k_{\rm max} - k_{\rm min} =  k(\alpha_{\rm min}) - k(\alpha_{\rm max})$, 
    \begin{lemma}
    \label{lemma:k_inverse_func}
    The inverse function of $k(\alpha)$, denoted as $k^{-1}(\alpha)$, can be written as
    \begin{equation}
        \label{eq:k_inv}
        k^{-1}(\alpha) \triangleq \sqrt{\frac{1}{k^2(\alpha) +1 }}, \quad \alpha \in (0,1].
    \end{equation}
    \end{lemma}
    \begin{proof}
        See Appendix \ref{apdx:proof_of_k_inv}.
    \end{proof}
    Thus, using \eqref{eq:k_inv}, we can express $\alpha_{\rm min} \in (0,1]$ such that \eqref{eq:inv_condition_2} is satisfied as a function of $\alpha_{\text{max}} \in (0,1]$ as
    \begin{subequations}
    \label{eq:alpha_min_v_alpha_max}
        \begin{align}
            \alphamin(\alphamax) = k^{-1}(k(\alpha_{\text{max}}) + \Delta k_{\text{max}}),\\
            = k^{-1}\bigg(
                k(\alphamax) + \left\|{\mathbf{M}^{-1}}\right\|_2^{-1}  \left\|{\dpdth}\right\|_2^{-1}
            \bigg),
        \end{align}
    \end{subequations}
    where $\Delta k_{\rm max}$ is the upper bound of \eqref{eq:inv_condition_2} computed at the operating point. Note that if $\alpha_{\text{max}}=1$, i.e., we set a bus to have unity power factor, then $k(\alpha_{\rm max})$ is zero. 
    
    In Fig.~\ref{fig:radial_pfmin_v_pfmax}, we plot the expression for $\alpha_{\rm min}$, \eqref{eq:alpha_min_v_alpha_max}, as a function of $\alpha_{\rm max}$ for radial test cases, and the same is done for meshed test cases in Fig.~\ref{fig:mesh_pfmin_v_pfmax}. This visualizes the conditions implied by Theorem \ref{thm1:suff_cond} for various \textsc{Matpower} \cite{zimmerman2010matpower} test cases at their default, instantaneous operating point. 
    \begin{figure}[t!]
        \centering
        \includegraphics[width=0.85\linewidth,keepaspectratio]{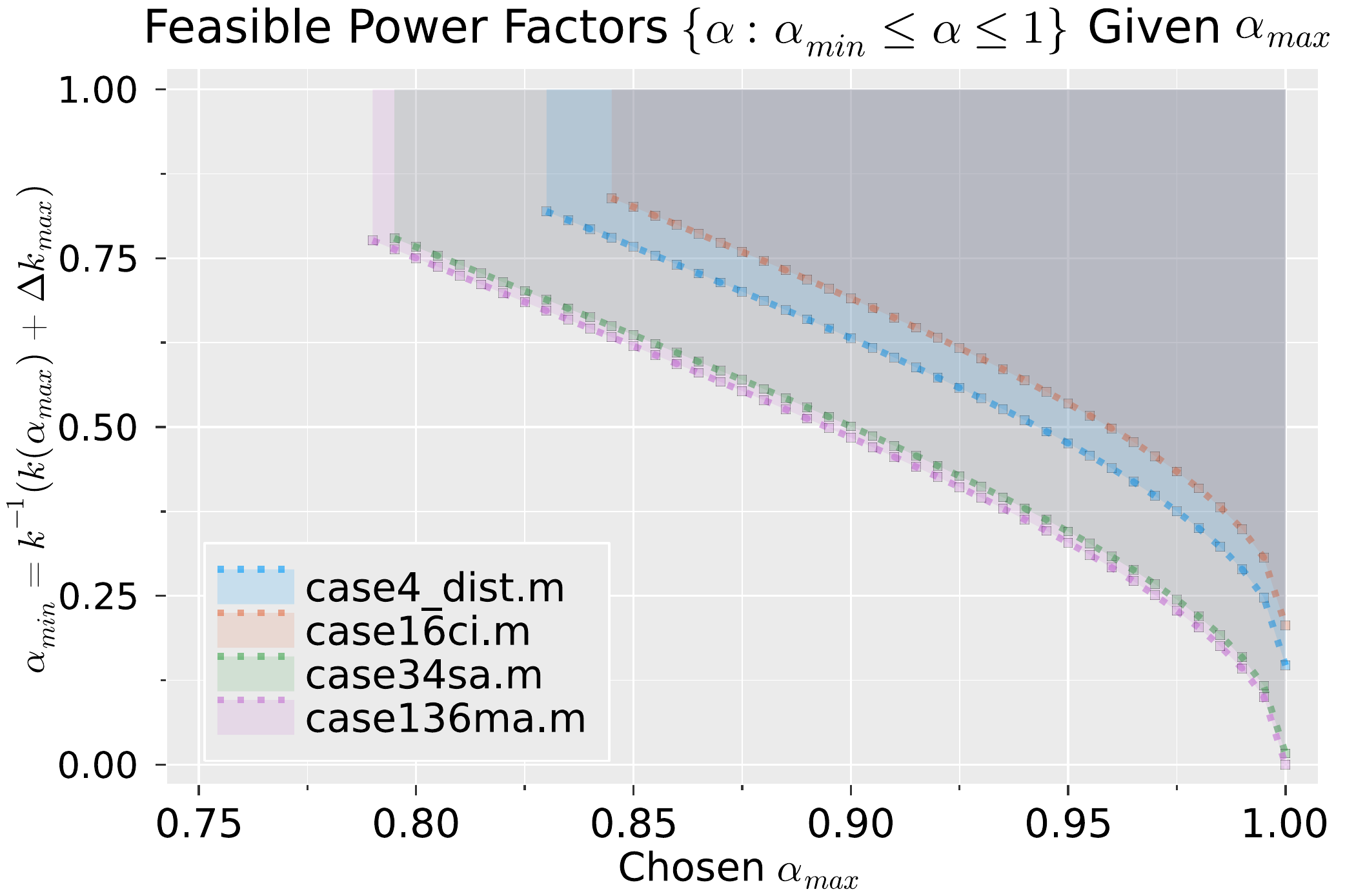}
        \caption{Radial cases: feasible bus power factors such that \eqref{eq:inv_condition_2} holds as a function of the maximum bus power factor $\alpha_{\rm max}$ using \eqref{eq:alpha_min_v_alpha_max}.}
        \label{fig:radial_pfmin_v_pfmax}
    \end{figure}
    \begin{figure}[t!]
        \centering
        \includegraphics[width=0.85\linewidth,keepaspectratio]{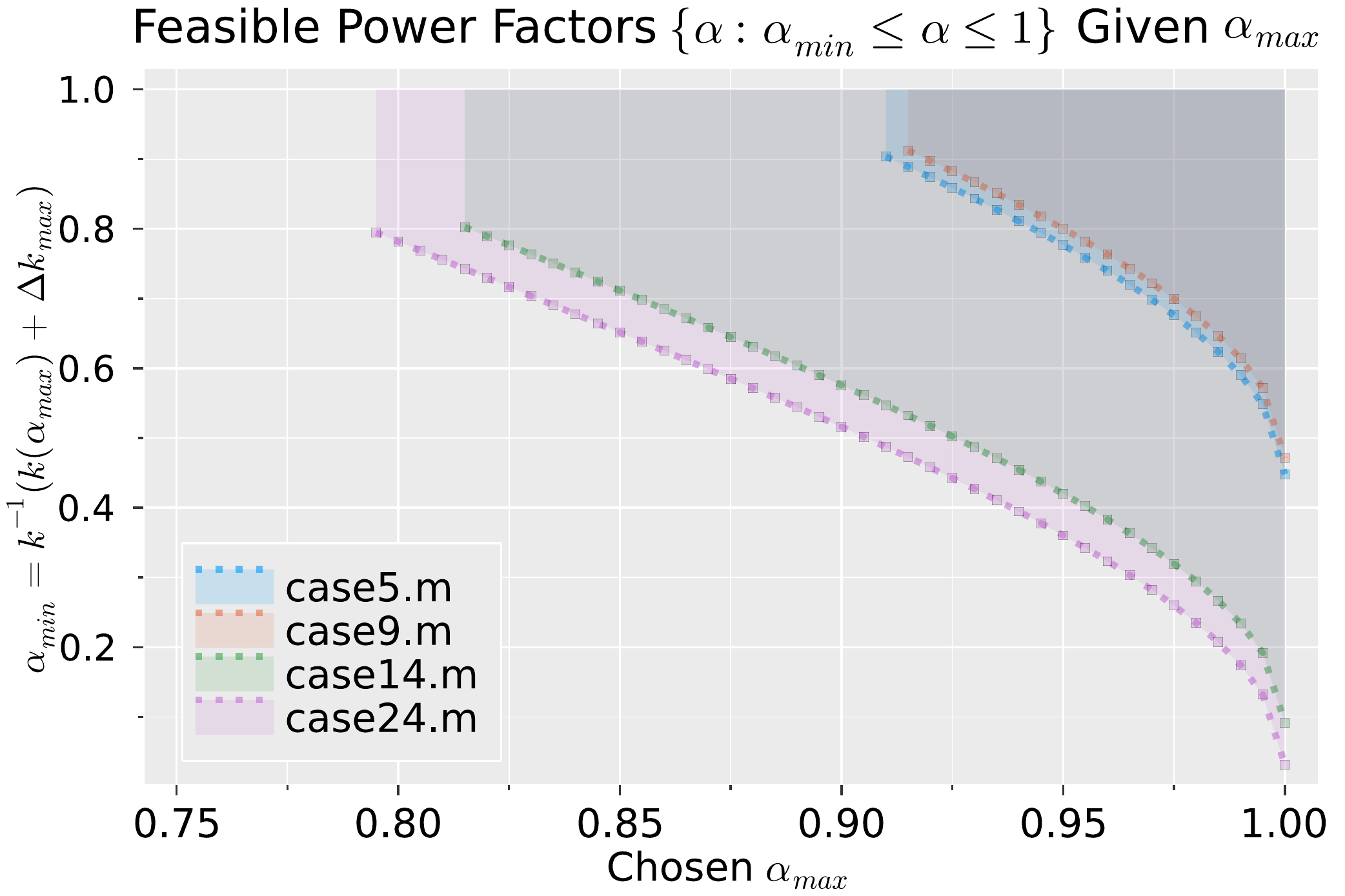}
        \caption{Meshed cases: feasible bus power factors such that \eqref{eq:inv_condition_2} holds as a function of the maximum bus power factor $\alpha_{\rm max}$ using \eqref{eq:alpha_min_v_alpha_max}.}
        \label{fig:mesh_pfmin_v_pfmax}
    \end{figure}
    
    
    
    
    \subsection{Necessary and sufficient condition}
    \label{sec:thm2_eigval}
         We propose the term \emph{phaseless observability} for the concept that is studied in this paper\textemdash a favorable grid operating condition where active and reactive power injection state estimates are well-conditioned from observations of the voltage magnitudes, \emph{without observing the phase angles.}
         
         Theorem \ref{thm1:suff_cond} establishes a \emph{sufficient} condition (bound) on the injection variables that, \emph{if} satisfied, guarantees that phaseless observability will hold in the network. In this section, we use bounds on the eigenvalues of linear operators derived from the Newton-Raphson power flow model \eqref{eq:classical_NR_iteration} that are both \textit{necessary and sufficient} for the phaseless observability condition to hold. 
         
         Since Theorem \ref{thm:eigenvalue_checks} is satisfied \emph{if and only if} phaseless observability holds, we propose that Theorem \ref{thm:eigenvalue_checks} captures the full scope of the property loosely characterized by the sufficient condition in Theorem \ref{thm1:suff_cond}. 
        
        Nonetheless, we note that the sufficient conditions developed in Section \ref{sec:thm1_neumann} remain valuable because of their intuitive connections to the bus power factors and the power flow Jacobian phase angle sensitivity submatrices. The robustness of the sufficient condition in the numerical results relative to the necessary and sufficient condition, coupled with the physical interpretation, motivates our elevation of Theorem \ref{thm1:suff_cond} as a primary contribution of our paper.
        
        This section introduces a stronger necessary and sufficient condition for the existence of a unique complex power injection solution in rectangular coordinates $\Delta \vx \in \R^{2n}$ given a wide voltage magnitude sensitivity matrix $\mS \in \R^{n \times 2n}$ and measurements of voltage magnitude perturbations $\Delta \vv \in \R^n$.  
        
        While this condition may have less physical interpretability than the sufficient conditions of Theorem \ref{thm1:suff_cond}, it remains valuable from both a theoretical and engineering perspective for evaluating the ``tightness" of the sufficient conditions in Section \ref{sec:thm1_neumann}, which we later include in the numerical results.
        
            
        
            
            The value of understanding these physical conditions is that it allows us to take a principled, physics-informed approached to estimating power factor parameters, through the form of a matrix $\mK$ described in Lemma \ref{lemma:implicit_representation}. 
        The engineering value of $\mK$ is that this matrix allows us to examine the eigenvalues of special \emph{projected} forms of the voltage sensitivity matrices\textemdash to be defined in Theorem \ref{thm:eigenvalue_checks}\textemdash to verify an estimate existence guarantee that is analogous to the ones provided by the conditions \eqref{eq:inv_condition_1} and \eqref{eq:inv_condition_2} of Theorem \ref{thm1:suff_cond}.

        
        \begin{theorem}[Phaseless observability]
        \label{thm:eigenvalue_checks}
        Let $\Delta \boldsymbol{v} \triangleq \vv - \vv_0 \in \mathbb{R}^n$ be a vector of voltage magnitude perturbations around an operating point $\vv_0 \in \R^n$ and let $\valpha \in (0,1)^n \subset \R^n$ be the bus power factors. Define the $n\times n$ special sensitivity matrices $ \stildedag(\valpha) : \R^n \mapsto \R^{n \times n} $ and $\stildestar(\valpha): \R^n \mapsto \R^{n \times n}$ as \eqref{eq:encoded-matrices}:
        \begin{subequations}
        \label{eq:encoded-matrices}
        \begin{align}
           \stildedag(\valpha) \triangleq \left(\dvdp + \dvdq \mathbf{K}(\valpha) \right),\\
            \stildestar(\valpha) \triangleq \left( \dvdp \mathbf{K}^{-1}(\valpha) + \dvdq \right),
        \end{align}
        \end{subequations}
        such that 
        \begin{equation}
        \Delta \vv(\valpha) = \stildedag(\valpha) \Delta \vp = \stildestar(\valpha) \Delta \vq.
        \end{equation}
    Then, there are unique complex power perturbations in rectangular coordinates $\Delta \boldsymbol{x} = [\Delta \vp^T, \Delta \vq^T]^T \in \R^{2n}$ such that $\Delta \boldsymbol{v} = \tilde{\mathbf{S}} \Delta \boldsymbol{x}, \ \tilde{\mS} \in \R^{n \times 2n},$ if and only if
        \begin{equation}
            \label{eq:thm2_nonsing_condition}
             \stildedag(\valpha) \succ 0.
        \end{equation}
        
    
    \end{theorem}
    \begin{proof}
        If $\stildedag \succ 0$ or $\stildestar \succ 0$ the linear systems of equations
        \begin{align}
            \stildedag \vxbar = \lambda \vxbar = \vct{0},\\
            \stildestar \vxbar' = \lambda \vxbar' = \vct{0},
        \end{align} 
        have solutions $\vxbar,\vxbar' \in \R^{n}$ if and only if $\vxbar,\vxbar' = \vct{0}$. If $\stildedag \succ 0$, then $\stildestar \triangleq \stildedag \mK^{-1}$ must be invertible as well, given the assumption that $\valpha \in (0,1)^n$.

    \end{proof}
    In the numerical results developed in Section \ref{sec:case_studies}, we will compare the application of Theorem \ref{thm1:suff_cond} and Theorem \ref{thm:eigenvalue_checks} to many test cases.

    
        

\section{Algorithms and Applications}
\label{sec:algorithms_application}
This section develops algorithms to reconstruct or update an estimate $\tilde{\mathbf{S}}^{\#}$ of the wide voltage magnitude-power sensitivity matrix $\tilde{\mathbf{S}}$ that we have studied in this paper, defined in  \eqref{eq:underdetermined_nr_inv_system}. Since we will assume that we do not have access to a network model, we use finite differences of the signals in \eqref{ami_signal} as the data inputs to these algorithms. Let $m'\triangleq m-1$ be the number of finite differences, and define  $\Delta \mathbf{V},\Delta \mathbf{P},\Delta \mathbf{Q} \in \mathbb{R}^{m' \times n}$ as matrices whose rows are the transpose of the finite difference vectors. Define a matrix of complex power perturbations in rectangular coordinates as $\Delta \mathbf{X} \triangleq [\Delta \mathbf{P}^T,\Delta \mathbf{Q}^T]^T \in \mathbb{R}^{m' \times 2n}$.



\subsection{Review of Least Squares Sensitivity Matrix Estimation}
\label{sec:least_squares}
Assuming $m' \geq 2n$, it is well-known \cite{nowak_measurement-based_2020,mugnier_model-lessmeasurement-based_2016} that a least-squares estimate for $\tilde{\mathbf{S}}$ can be found via the Moore-Penrose Pseudoinverse as $\big(\tilde{\mathbf{S}}^{\#}\big)^T = (\Delta \mathbf{X}^T\Delta \mathbf{X} + \lambda\mathbf{I})^{-1} \Delta \mathbf{X} ^T \Delta \mathbf{V}$, where $\lambda$ is a Tikhonov regularization parameter. The resulting estimate gives us $\Delta \mathbf{V}^T \approx \tilde{\mathbf{S}}^{\#} \Delta \mathbf{X}^T$.




We propose that the wide $\tilde{\mathbf{S}}$ will often have rapidly decreasing singular values. In this case, $\tilde{\mathbf{S}}$ can be well approximated via a truncated singular value decomposition (SVD) as:
\begin{equation}
\label{truncated_svd}
    \tilde{\mathbf{S}} \approx \sum_{k=1}^R \sigma_k \boldsymbol{u}_k \boldsymbol{v}_k^T,
\end{equation}
where $\sigma_k,\boldsymbol{u}_k,$ and $\boldsymbol{v}_k$, $k=1,\dots,R$ are the $R$ largest singular values and corresponding singular vectors. The assumption of rapidly decreasing singular values is well-motivated, as can be verified empirically in Fig. \ref{fig:spectral_analysis_ieee13}, which shows a spectral analysis of the voltage sensitivities for the IEEE 13-bus test feeder.

\begin{remark}
\label{remark:approx_low_rank}
The approximate low-rank structure of $\tilde{\mathbf {S}}$ results from the columns belonging to a union of low-rank subspaces. Empirically, we have found these are related to groupings of the injection type (P/Q) and phase (A/B/C).
\end{remark}

\subsection{Sensitivity Matrix Completion}
\label{sec:matrix_completion_update}
Suppose that we have an incomplete sensitivity matrix $\tilde{\mathbf {S}}_0 \triangleq [\mathbf{S}^v_{p,0},\mathbf{S}^v_{q,0}]$ where the set $\Omega = \{(i,j) : [\tilde{\mathbf {S}}_0]_{i,j} = 0\}$ represents $|\Omega|$ entries of $\tilde{\mathbf{S}}_0$ for which we do not have access to voltage sensitivity relationships. 
The full matrix $\tilde{\mathbf{S}}$, which contains entries for all buses, can be recovered as the solution to the following program:
\begin{equation}
    \tilde{\mathbf S}^\# = \mathop{\rm arg \, min}_{{\mathbf {S}} \in \mathbb{R}^{n \times 2n}} ||\tilde{\mathbf {S}}_0 - {\mathbf {S}}||_F^2 \quad \text{subject to:} \quad \text{rank}({\mathbf {S}}) = R,
\label{matrix_completion}
\end{equation}
where $||\cdot||^2_F$ is the squared Frobenius norm, which is defined for a matrix $\mathbf{X} \in \mathbb{R}^{d_1 \times d_2}$ as $||\mathbf{X}||_F^2 = \sum_{i=1}^{d_1} \sum_{j=1}^{d_2} |X_{i,j}|^2$. 

\begin{figure*}[htb]
\centerline{\includegraphics[width=0.975\linewidth,height=.25\textheight,keepaspectratio]{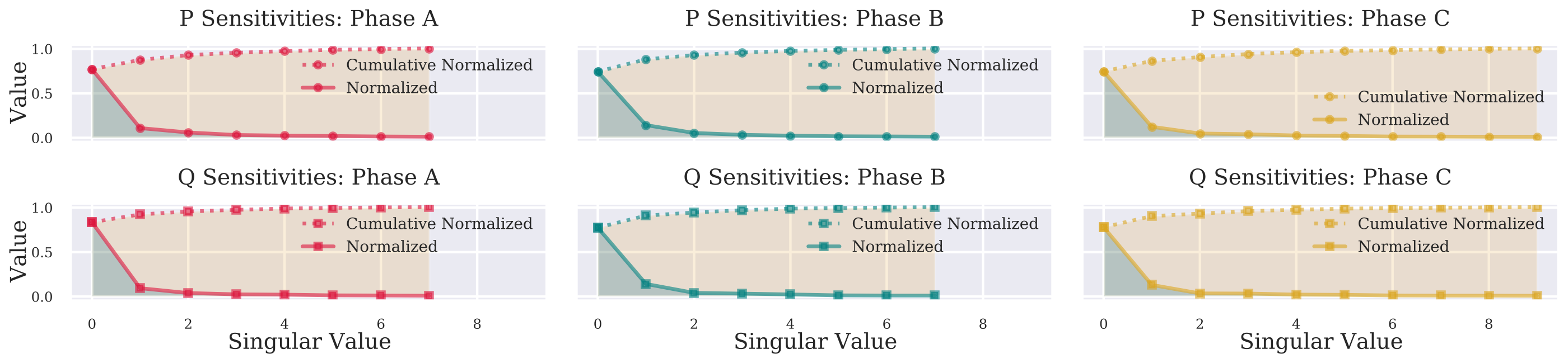}}
    \caption{Spectral analysis of the $\tilde{\mathbf{S}}$ matrix \eqref{eq:underdetermined_nr_inv_system} by phase and injection type for the IEEE 13-bus test case, showing approximate low-rank structure.}
    \label{fig:spectral_analysis_ieee13}
\end{figure*}

The program \eqref{matrix_completion} is non-convex, but a closed form solution can be tractably found by truncating the SVD  as in \eqref{truncated_svd}. Choosing $R$ is equivalent to tuning a real-valued hyperparameter  $\lambda \geq 0$ in the Lagrangian of this program,
\begin{equation}
    \label{matrix_recovery_problem_lagrangian}
    \mathop{\rm arg \, min}_{{\mathbf {S}}} 
    {||\tilde{\mathbf {S}}_0- {\mathbf {S}}||_F^2 + \lambda \left( \text{rank}({\mathbf {S}}) \right)}.
\end{equation}
The rank constraint on the optimization variable $\mathbf{S}$ is also non-convex, and the solution requires \textit{hard-thresholding}, i.e., selecting an integer $R$ in \eqref{truncated_svd}. Additionally, we cannot solve \eqref{matrix_recovery_problem_lagrangian} in this way, as we cannot take the truncated SVD of a matrix with unknown values. Following \cite{yao_nuclear_norm,davenport_overview_2016,ospina_sensitivity_estimation}, this leads to the convex relaxation \eqref{matrix_recovery_nuclear_problem_lagrangian}, which replaces the rank penalty term with the nuclear norm of the decision matrix:
\begin{equation}
\label{matrix_recovery_nuclear_problem_lagrangian}
\begin{split}
    &\mathop{\rm arg \, min}_{{\mathbf {S}}} 
    {||\tilde{\mathbf {S}}_0- {\mathbf {S}}||_F^2 + \lambda ||{\mathbf {S}}||_*},\\
    &\quad \text{s.t.} \quad ||\mathbf{S} - {\mathbf {S}}_{\Omega}||_{F} \leq \delta,
\end{split}
\end{equation}
where $[{\mathbf{S}}_{\Omega}]_{i,j} = 0 \ \forall (i,j) \in \Omega$, and $[{\mathbf{S}}_{\Omega}]_{i,j} = [\tilde{\mathbf{S}}_0]_{i,j}$ if $(i,j) \notin \Omega$. The operator $||\cdot||_*$ denotes the nuclear norm, which is the sum of the singular values of ${\mathbf {S}}$. 
The hyperparameter $\delta$ reflects how accurately we wish to match the coefficients that are known beforehand in $\tilde{\mathbf {S}}_0$. The program \eqref{matrix_recovery_problem_lagrangian} promotes solutions with skewed singular values, which are ``approximately low-rank". 

\subsection{Iterative Estimation and Completion}
In contrast with the well-studied least-squares method in Section \ref{sec:least_squares}, in this section we assume that  $m' \ll 2n$, and that we have access to a number of precomputed local sensitivity coefficients.
In this setting, given a small chunk of the finite differences of the AMI measurements described at the beginning of this section, we can solve
    \begin{subequations}
    \label{eq:partial_information_nuc_reg_lsq}
    \begin{align}
        \tilde{\mathbf {S}}^{\#}_t = & \argmin_{\mathbf{S}} \| \mathbf{S} \Delta \mathbf{X} - \Delta \mathbf{V} \|_F^2 + \lambda \| \mathbf{S} \|_*,\\
        & \text{subject to:} \quad \| \mathbf{S}_{\Omega} - \mathbf{S}\|_F \leq \delta,
    \end{align} 
    \end{subequations}
or, alternatively, we can use the measurements in \eqref{ami_signal} sequentially to perform a similar iterative estimation of the sensitivity matrix at time $t$, $\tilde{\mathbf {S}}^{\#}_t$, by solving the online convex optimization problem \eqref{eq:reg_dyn_problem}:
{\small
\begin{align}\nonumber
    \tilde{\mathbf {S}}^{\#}_t = & \mathop{\rm arg \, min}_{{\mathbf{S}}} ||\Delta \boldsymbol{v}_t- {\mathbf S} \Delta {\boldsymbol x}_t||_2^2 + \lambda ||{\mathbf {S}}||_*  + c \sum_{s=1}^{t-1} \gamma^{s} || \tilde{\mathbf {S}}^{\#}_{t-s} - {\mathbf {S}}||_F^2, \\
    & \quad\; \text{subject to:} \quad ||\mathbf{S} - {\mathbf {S}_{\Omega}}||_{F} \leq \delta. \label{eq:reg_dyn_problem}
\end{align}
}%
The summation term in the optimization is a penalty term: if we consider $\tilde{\mathbf {S}}^{\#}_t$ for all $t$ as a time series, then the summation is equivalent to an exponential smoother. The time constant of the smoother is $\gamma \in (0,1)$ and the strength of this penalty term is given by the hyperparameter $c$. The purpose of this term is to smooth out any sharp difference between the various $\tilde{\mathbf {S}}^{\#}_t$ at contiguous time steps. The voltage and power perturbations at time $t$ are the vectors $\Delta {\boldsymbol v}_t \in \mathbb{R}^{n}, \ {\boldsymbol x}_t \in \mathbb{R}^{2 n}$. 

\section{Case Studies}
\label{sec:case_studies}
    This section provides numerical case studies of the theory and algorithms developed in this paper. In Section \ref{sec:misallignments_theory}, we outline the preprocessing steps we use. In Section \ref{sec:comp_test_power_factor_bound}, we compute the analytical upper bound on $\Delta k$ derived in \eqref{eq:inv_condition_2} of Theorem \ref{thm1:suff_cond} and test the validity of Assumption \ref{assum:nonsingular_jacobian} for numerous radial and meshed networks. This is done in \texttt{PowerModels.jl} \cite{powermodels} by using and extending the \texttt{{calc\_basic\_jacobian\_matrix}} function, the results of which are shown in Table \ref{table:invertibility_bound_radial} and Table \ref{table:invertibility_bound_meshed}. In Section \ref{sec:sens_matrix_recovery}, we apply the estimation techniques developed in Section \ref{sec:algorithms_application} to the IEEE 13-bus and 123-bus radial distribution test cases \cite{kersting_radial_1991} using the \texttt{OpenDSS} distribution network simulator \cite{dugan_opendss}.

 \subsection{Misalignments with Theoretical Assumptions}
    \label{sec:misallignments_theory}
    To test Theorem \ref{thm1:suff_cond} numerically in Section \ref{sec:comp_test_power_factor_bound}, we take the following practical preprocessing steps to generate the results in Table \ref{table:invertibility_bound_radial} and Table \ref{table:invertibility_bound_meshed}. We study buses that:
    \begin{enumerate}
        \item are a PQ bus,
        \item have a nonzero net active and apparent power injection.
    \end{enumerate}
    For the results generated in Table \ref{table:invertibility_bound_radial} and Table \ref{table:invertibility_bound_meshed}, we assume that buses with zero net power injection within a tolerance of $\epsilon=1\times10^{-6}$ will have the corresponding entry in the matrix $\mathbf{K}$ at this bus replaced with the sample mean of the nonzero elements of $\operatorname{diag}(\mathbf{K})$. 

\subsection{Test of Analytical Results}
    \label{sec:comp_test_power_factor_bound}
    In the results of this computation, we maintain the default operating points specified in the network data.
    
    \subsubsection{Theorem 1}
    In Table \ref{table:invertibility_bound_radial} and Table \ref{table:invertibility_bound_meshed}, the quantity $\norm{\mathbf{M}^{-1} \Delta \mathbf{K} \dpdth}_2$ is shown for numerous radial and meshed \textsc{Matpower} test cases, respectively. The quantity $\norm{\mathbf{M}^{-1} \Delta \mathbf{K} \dpdth}_2$ must be strictly less than 1 for the sufficient condition \eqref{eq:inv_condition_1} to hold, which implies that there is a unique estimate for complex power perturbations from the voltage magnitudes. We also report the stricter bound, $\norm{\mathbf{M}^{-1}}_2^{-1} \norm{\dpdth}^{-1}_2$, which is useful for its physical interpretation via the bus power factors. 
   Note that Theorem \ref{thm1:suff_cond} holds for all case 5 variants provided by \texttt{PowerModels.jl} except for: 1.) \texttt{case5\_db}, as its Jacobian is singular, as well as 2.) \texttt{case5\_sw} and 3.) \texttt{case5\_tnep}, as they have a single PQ bus.
 \begin{table*}[htb]
        \renewcommand{\arraystretch}{1.4}
        \caption{Numerical validation of Theorems for radial \textsc{Matpower} test cases at default operating point}
         \label{table:invertibility_bound_radial}
        \tabcolsep=0.11cm
      
        \vspace{-3mm}
        \begin{center}
        \begin{tabular}{|c||c||c||c||c||c||c||c||c|c|}
        \hline
          & \multicolumn{2}{c||}{\textbf{Assumption Holds?}} &  \multicolumn{5}{c||}{\textbf{Quantity}}  &  \multicolumn{2}{c|}{\multirow{1}{*}{\textbf{Thm. Holds?}}}\\
        \cline{1-10}
        \textbf{Case } 
        &  $\dpdth \succ 0$ 
        & $\mathbf{J}^{-1}$  Exists 
        & $\lambda_{\rm min}(\mathbf{J})$
        & $\alpha_{\text{max}}- \alpha_{\text{min}}$ 
        & $k_{\text{max}}-k_{\text{min}}$ 
        & $\|\mathbf{M}^{-1} \|^{-1}_2 \big\lVert \dpdth \big\rVert_2^{-1}$ 
        & $\big \lVert \mathbf{M}^{-1} \Delta \mathbf{K} \dpdth \big \rVert_2$ 
        & Thm. \ref{thm1:suff_cond} 
        & Thm. \ref{thm:eigenvalue_checks}\\
        \hline
        2 & Yes & Yes &   17.3726   &  0.0 & 0.0 & 0.403 & 0.0 & \textcolor{ForestGreen}{Yes} & \textcolor{ForestGreen}{Yes}\\
        \hline
        5\_tnep & Yes & Yes &  104.01    &  0.0 & 0.0 & 0.4553 & 0.0 & \textcolor{ForestGreen}{Yes} & \textcolor{ForestGreen}{Yes}\\
        \hline
        4\_dist & Yes & Yes &   50.426   &  0.0 & 0.0 & 0.1472 & 0.0 & \textcolor{ForestGreen}{Yes} & \textcolor{ForestGreen}{Yes}\\
        \hline
        10ba & Yes & Yes &     0.646  &  0.34293 & 2.833 & $6.596 \times 10^{-3}$ & 6.21 & \textcolor{Sepia}{No} & \textcolor{Sepia}{No}\\
        \hline
        12da & Yes & Yes &    0.322   &  0.0929 & 0.250  & 0.01072 & 5.38 & \textcolor{Sepia}{No} & \textcolor{ForestGreen}{Yes}\\
        \hline
        15da & Yes & Yes &   1.772   &  $7.14 \times 10^{-8}$ & $2.04 \times 10^{-7}$& 0.6299 & $1.74 \times 10^{-6}$ & \textcolor{ForestGreen}{Yes}
        & \textcolor{ForestGreen}{Yes}\\
        \hline
        15nbr  & Yes & Yes &  0.0262    &  $7.14 \times 10^{-8}$ & $2.04 \times 10^{-7}$ & 0.0156 & $1.76 \times 10^{-6}$ & \textcolor{ForestGreen}{Yes}
        & \textcolor{ForestGreen}{Yes}\\
        \hline
        16am   & Yes & Yes &   6.117   &  0.248 & 0.767 & 0.198 & $0.993$ & \textcolor{ForestGreen}{Yes}
        & \textcolor{ForestGreen}{Yes}\\
        \hline
        16ci  & Yes & Yes &   9.968   &  0.198 & 0.54 & 0.206 & $0.40$ & \textcolor{ForestGreen}{Yes}
        & \textcolor{ForestGreen}{Yes}\\
        \hline
        17me   & Yes & Yes &   0.0651   &  0.248 & 0.767 & 0.01204 & $17.55$& \textcolor{Sepia}{No}
        & \textcolor{ForestGreen}{Yes}\\
        \hline
        18     & Yes & Yes &    0.542  &  0.0588 & 0.1533 & 0.00573 & $3.0485$& \textcolor{Sepia}{No} & \textcolor{ForestGreen}{Yes}\\
        \hline
        18nbr & Yes & Yes &    0.0189  &   $1.43\times 10^{-4}$  & $4.08 \times 10^{-4}$ & 0.0115 & $6.2 \times 10^{-3}$ & \textcolor{ForestGreen}{Yes} & \textcolor{ForestGreen}{Yes}\\
        \hline
        22     & Yes & Yes &   1.0580    &  $0.164$    &  $0.495$      & $7.182 \times 10^{-3}$ & $6.39$& \textcolor{Sepia}{No} & \textcolor{ForestGreen}{Yes}\\
        \hline
        28da   & Yes & Yes &   0.274 &  $6.25 \times 10^{-6}$ & $1.79 \times 10^{-5}$ & $2.157 \times 10^{-3}$ & $5.0 \times 10^{-4}$ & \textcolor{ForestGreen}{Yes} & \textcolor{ForestGreen}{Yes}\\
        \hline
        33bw   & Yes & Yes &    0.0870   &  $0.670$ & $2.833$ & $2.732 \times 10^{-3}$ & $27.82$& \textcolor{Sepia}{No} & \textcolor{Sepia}{No}\\
        \hline
        33mg  & Yes & Yes &   0.784   &  $0.670$  & $2.83$ & $2.50 \times 10^{-3}$ & $28.024$& \textcolor{Sepia}{No} & \textcolor{ForestGreen}{Yes}\\
        \hline
        34sa  & Yes & Yes &   1.120    &  $0.0534$  & $0.238$ & 0.0166 & $0.573$ & \textcolor{ForestGreen}{Yes} & \textcolor{ForestGreen}{Yes}\\
        \hline
        38si   & Yes & Yes &   0.708   &   $0.670$  & $2.833$ & 0.160 & $17.792$& \textcolor{Sepia}{No} & \textcolor{ForestGreen}{Yes}\\
        \hline
        51ga   & Yes & Yes &  0.575  &  $0.251$  & $0.701$ & $1.105 \times 10^{-3}$ & $15.484$& \textcolor{Sepia}{No} & \textcolor{Sepia}{No}\\
        \hline
        51he   & Yes & Yes &   1.552    &  $0.119$  & $0.321$ & 0.0112 & $1.720$& \textcolor{Sepia}{No} & \textcolor{Sepia}{No}\\
        \hline
        69     & Yes & Yes &    0.0489   &  $0.100$   & $0.263$ & $0.123 \times 10^{-3}$ & $0.518$ & \textcolor{ForestGreen}{Yes} & \textcolor{ForestGreen}{Yes}\\
        \hline
        70da    & Yes & Yes &   0.737   &   $0.194$   & $0.523$ & $2.879 \times 10^{-3}$ & $9.820$& \textcolor{Sepia}{No} & \textcolor{ForestGreen}{Yes}\\
        \hline
        74ds    & Yes & Yes &     1.0596  &  $0.161$   & $0.429$ & $0.344 \times 10^{-3}$ & $30.25$& \textcolor{Sepia}{No} & \textcolor{ForestGreen}{Yes}\\
        \hline
        85     & Yes & Yes &   0.1835    &  $1.25 \times 10^{-6}$ & $3.57 \times 10^{-6}$   & 0.0191 & $1.28 \times 10^{-5}$ & \textcolor{ForestGreen}{Yes} & \textcolor{ForestGreen}{Yes}\\
        \hline
        94pi    & Yes & Yes &    0.2748   &  $5.86 \times 10^{-3}$ & $16.7 \times 10^{-3}$ & $1.082\times 10^{-3}$  & $0.0864$ & \textcolor{ForestGreen}{Yes} & \textcolor{ForestGreen}{Yes}\\
        \hline
        118zh   & Yes & Yes &    0.1438  &   $0.445$   & $1.412$ & $0.506 \times 10^{-3}$ & $22.53$& \textcolor{Sepia}{No} & \textcolor{Sepia}{No}\\
        \hline
        136ma   & Yes & Yes &   0.1854    &  $0.0309$ &  $0.09137$ & $0.135 \times 10^{-3}$ & $0.158$ & \textcolor{ForestGreen}{Yes} & \textcolor{ForestGreen}{Yes}\\
        \hline
        141   & Yes & Yes &    0.1067   &  $1.493 \times 10^{-5}$   & $3.92 \times 10^{-9}$ & $5.496 \times 10^{-6}$ & $5.48 \times 10^{-7}$ & \textcolor{ForestGreen}{Yes} & \textcolor{ForestGreen}{Yes}\\
        \hline
        
        \end{tabular}
        \end{center}
        \vspace*{-1em}
    \end{table*}
    \begin{table*}[ht]
        \renewcommand{\arraystretch}{1.4}
        \caption{Numerical validation of Theorems for meshed \textsc{Matpower} test cases at default operating point}
        \label{table:invertibility_bound_meshed}
        \vspace{-3mm}
        \tabcolsep=0.11cm
        \begin{center}
        \begin{tabular}{|c||c||c||c||c||c||c||c||c|c|}
        \hline
          & \multicolumn{2}{c||}{\textbf{Assumption Holds?}} &  \multicolumn{5}{c||}{\textbf{Quantity}}  &  \multicolumn{2}{c|}{\multirow{1}{*}{\textbf{Thm. Holds?}}}\\
        \cline{1-10}
        \textbf{Case } 
        &  $\dpdth \succ 0$ 
        & $\mathbf{J}^{-1}$ Exists 
        & $\lambda_{\rm min}(\mathbf{J})$
        & $\alpha_{\text{max}}- \alpha_{\text{min}}$ 
        & $k_{\text{max}}-k_{\text{min}}$
        & $\norm{\mathbf{M}^{-1}}^{-1}_2 \big\lVert \dpdth \big\rVert_2^{-1}$ 
        & $\|\mathbf{M}^{-1} \Delta \mathbf{K} \dpdth \|_2$ 
         & Thm. \ref{thm1:suff_cond} 
        & Thm. \ref{thm:eigenvalue_checks}\\
        \hline
        case5 & Yes & Yes& 1.0 & 0.0 & 0.0 & 0.448 & 0.0 & \textcolor{ForestGreen}{Yes} & \textcolor{ForestGreen}{Yes}\\
        \hline
        case9 & Yes & Yes&   0.766   &  0.0563 & 0.194 & 0.471 & 0.280 & \textcolor{ForestGreen}{Yes} & \textcolor{ForestGreen}{Yes}\\
        \hline
        case14 & Yes & Yes&   0.549   &   0.138 & 0.434 & 0.0915 & 0.474 & \textcolor{ForestGreen}{Yes} & \textcolor{ForestGreen}{Yes}\\
        \hline
        case24 & Yes & Yes& 1.0     &   $1.65 \times 10^{-3}$ & $8.70\times 10^{-3}$ & $0.0317$ & 0.0230 & \textcolor{ForestGreen}{Yes} & \textcolor{ForestGreen}{Yes}\\
        \hline
        case30 & Yes & Yes&    0.235  &   0.192 & 0.591 & 0.1472 & 1.335 & \textcolor{Sepia}{No} & \textcolor{ForestGreen}{Yes}\\
        \hline
        \end{tabular}
        \end{center}
    \end{table*}

    Tables \ref{table:invertibility_bound_radial} and \ref{table:invertibility_bound_meshed} verify that when all buses in a network have constant, nonunity power factors, as in \texttt{case4\_dist}, which has $\alpha_i = 0.894$ for all buses $i$, the condition of Theorem \ref{thm1:suff_cond} is trivially satisfied. The results also verify that the sufficient condition \eqref{eq:inv_condition_1} and the stronger, physically interpretable condition \eqref{eq:inv_condition_2} are satisfied for many test cases with differing, non-unity power factors at their default operating points. This is observed for both radial and meshed cases, as shown in Table \ref{table:invertibility_bound_radial} and \ref{table:invertibility_bound_meshed}, respectively.  
    
    Additionally, cases that do not satisfy the condition typically have large variations between the bus power factors. Therefore, we hypothesize that future research could potentially leverage Theorem \ref{thm1:suff_cond} to design control algorithms, such as those that follow the formulation of \cite{baker_network-cognizant_2018}. Another future research direction is the design of power factor controller that manages the net injection power factors of the loads such that the complex power injections remain observable from the voltage magnitude deviations. We propose that this would have applications in distribution grid sensor placement and expansion planning problems, where the costs and benefits of increasing penetration of PMUs must be considered. 
    


    \subsubsection{Theorem 2}
    \label{sec:case-study-thm2}
      Results for Theorem \ref{thm:eigenvalue_checks} are also listed in Table \ref{table:invertibility_bound_radial} and Table \ref{table:invertibility_bound_meshed}. Additional  numerical results for testing Theorem \ref{thm:eigenvalue_checks} on the radial cases that were listed as satisfying Theorem \ref{thm:eigenvalue_checks} are provided in Table \ref{table:thm2_radial_test} in Appendix \ref{apdx:additional-numerical-results} Note that all $\mathbf{K}$ matrices discussed in Table \ref{table:thm2_radial_test} are nonsingular.

\subsection{Complex Power and Sensitivity Matrix Estimation}
\label{sec:sens_matrix_recovery}
In this section, we present two case studies for the voltage sensitivity matrix completion problems outlined in the second part of the paper, which can complement existing regression-based methods for estimating the matrices.
    
    \subsubsection{IEEE 13-Bus Test Case}
    We compute the voltage sensitivities to active and reactive power injections for the IEEE 13-bus test case using OpenDSS as a baseline for comparison. The default loadshape is used for all loads.  CVXPY \cite{diamond2016cvxpy} is used to implement the matrix completion algorithms. To represent a varying degree of sensor penetration, we change the number of observed sensitivity coefficients $|\Omega|$, from 20\% to 90\% of the total number of entries. We vary the nuclear norm penalty $\lambda$ between $1\times 10^{-6}$ and $8 \times 10^{-6}$. We fix $\delta = 6 \times 10^{-3}$. We reconstruct the active and reactive power sensitivities with a mean absolute error below $1.25 \times 10^{-6}$ for all sensor levels. 
    
    
    
    
    At the sensor observability level of 20\%, we used the online model update \eqref{eq:reg_dyn_problem} to estimate $\tilde{\mathbf{S}}^{\#}_t$ in real time. With a smoothing factor of $\gamma=0.9$, a nuclear norm penalty of $\lambda=1.25 \times 10^{-4}$, and a gain of $c=1 \times 10^{-8}$ for the smoothing term, we run the online optimization problem for 15-minute time steps at 10 different noise levels for the IEEE 13-bus test case, as shown in Fig.~\ref{online_results}. The errors are approximately an order of magnitude smaller than the values of $\norm{\Delta \boldsymbol{v}_t}$ and $\norm{\Delta\hat{\boldsymbol{v}}_t}$ themselves at all noise levels, which indicates the predictive performance of the method. 


    \begin{figure}[tb]
      \centerline{\includegraphics[width=0.8\linewidth,height=.3\textheight,keepaspectratio]{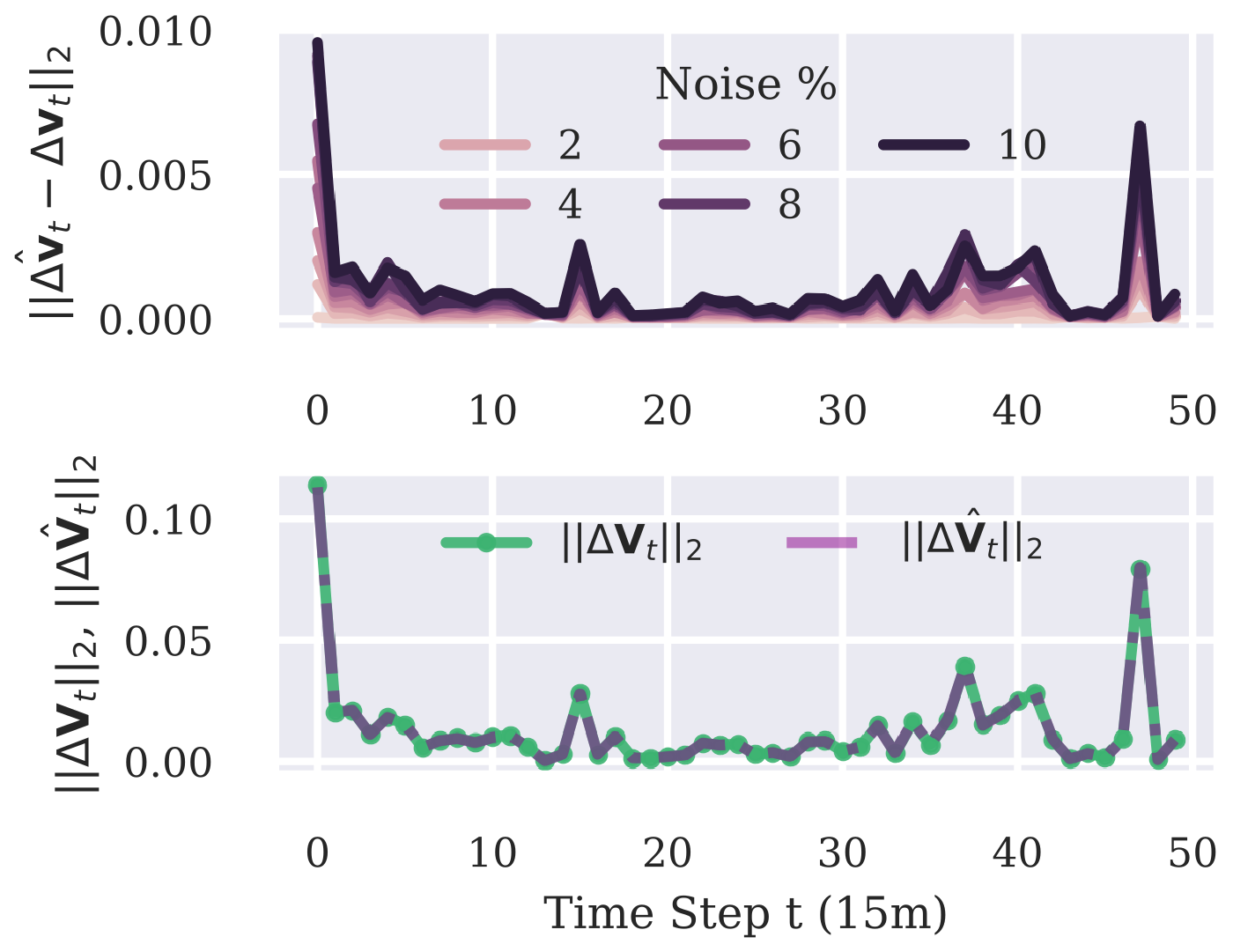}}
        \caption{Predictive performance of the $\tilde{\mathbf{S}}$ matrix found via \eqref{eq:reg_dyn_problem} for the IEEE 13-bus test feeder. The top figure shows the Euclidean distance vs. time between the observed and predicted voltage deviations at all buses for varying noise levels, shown as percentages of the mean measured voltage. The bottom figure shows the Euclidean norm vs. time of the predicted (purple) and the true (green) voltage deviations for all buses at the lowest noise level.}
        \label{online_results}
    \end{figure}
    
    \subsubsection{IEEE 123-Bus Test Case}
    This section extensively tests the methods developed in this paper on the IEEE 123-bus test feeder. We verify the reactive power representation in \eqref{eq:implicit_representation_def} through two load data inputs. First, we set all loads to fixed power factor control seeking to maintain a value of 0.9. Second, we allow the load power factor settings to vary over time between 0.8 and 0.9. The actual power factors reported by OpenDSS after simulation were 0.795 to 0.906. 
    
    For both data inputs, we compute a time-series of $\mathbf{S}^v_p,\mathbf{S}^v_q$ using the perturb-and-observe method, which computes these matrices by adding a small static active or reactive power injection iteratively to each bus and recording the normalized change in voltage magnitudes relative to the voltages at the base case solution. More information is available in \cite{chen_measurement-based_2016}. We then estimate $\Delta \boldsymbol{p}_t,\Delta \boldsymbol{q}_t$ using the voltage magnitudes and the $\stildedag$ and $\mathbf{K}$ matrices defined in \eqref{eq:implicit_representation_def} for $t=1,\dots,m'$ as
    \begin{equation}
        \label{eq:basic_estimation_computation}
        (\Delta \hat{\boldsymbol{p}}_t,\Delta \hat{\boldsymbol{q}}_t) = (\stildedag^{-1} \Delta \boldsymbol{v}_t, \mathbf{K} \stildedag^{-1} \Delta \boldsymbol{v}_t).
    \end{equation}

    The results of this computation are shown for 10 buses of the IEEE 123-bus test feeder in Fig. \ref{fig:implicit_complex_estimation_results_ieee123_constant_pf} and Fig. \ref{fig:implicit_complex_estimation_results_ieee123_vary_pf}. The overall root mean squared error (RMSE) of the estimation results over the 24-hour time horizon for both power factor data input scenarios are shown for all buses in Fig. \ref{fig:rmse_by_node}.
    
    \begin{figure}
        \centering
        \includegraphics[width=0.85\linewidth,keepaspectratio]{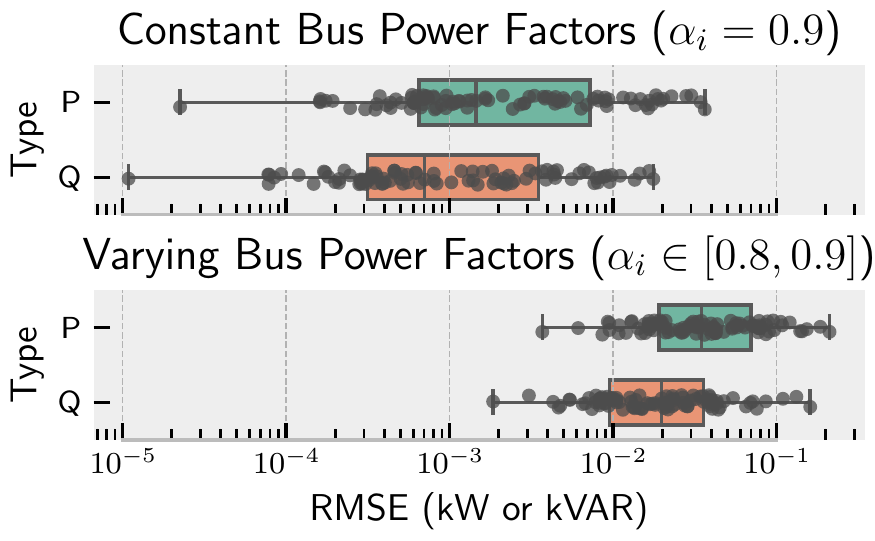}
        \caption{Root mean squared error (RMSE) of the active/reactive power deviation time-series estimated over a 24-hour time horizon for the IEEE 123-bus test case using \eqref{eq:basic_estimation_computation} with both constant (top) and varying (bottom) bus power factors. Each dot represents a single bus with a load.}
        \label{fig:rmse_by_node}
    \end{figure}
    
    \begin{figure*}
        \centering
        \includegraphics[width=0.985\linewidth,keepaspectratio]{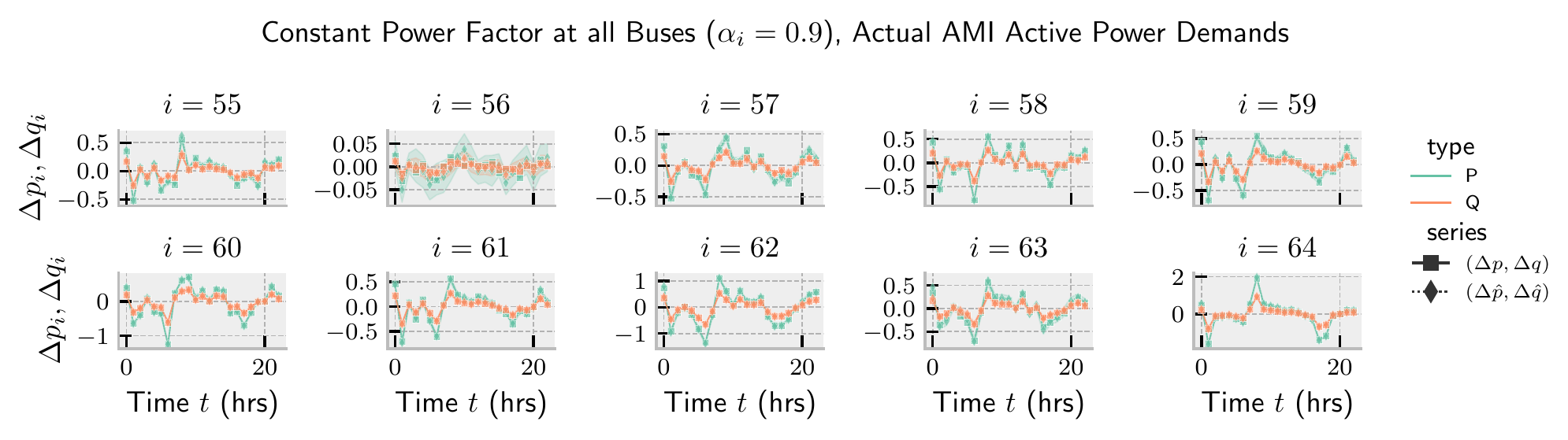}
         \vspace{-3mm}
        \caption{Estimating complex power using voltage magnitudes for the IEEE 123-bus case with fixed bus power factors of 0.9 using the basic implicit representation in \eqref{eq:basic_estimation_computation} from \eqref{eq:implicit_representation_def}. The measured perturbations are shown with squares/solid lines and the estimated values are shown with diamonds/dashed lines. The shaded regions denote a 95\% bootstrap CI found by repeatedly injecting noise into the AMI data with variance of 0.1\% of the mean power.}
        \label{fig:implicit_complex_estimation_results_ieee123_constant_pf}
    \end{figure*}
    \begin{figure*}
        \centering
        \includegraphics[width=0.985\linewidth,keepaspectratio]{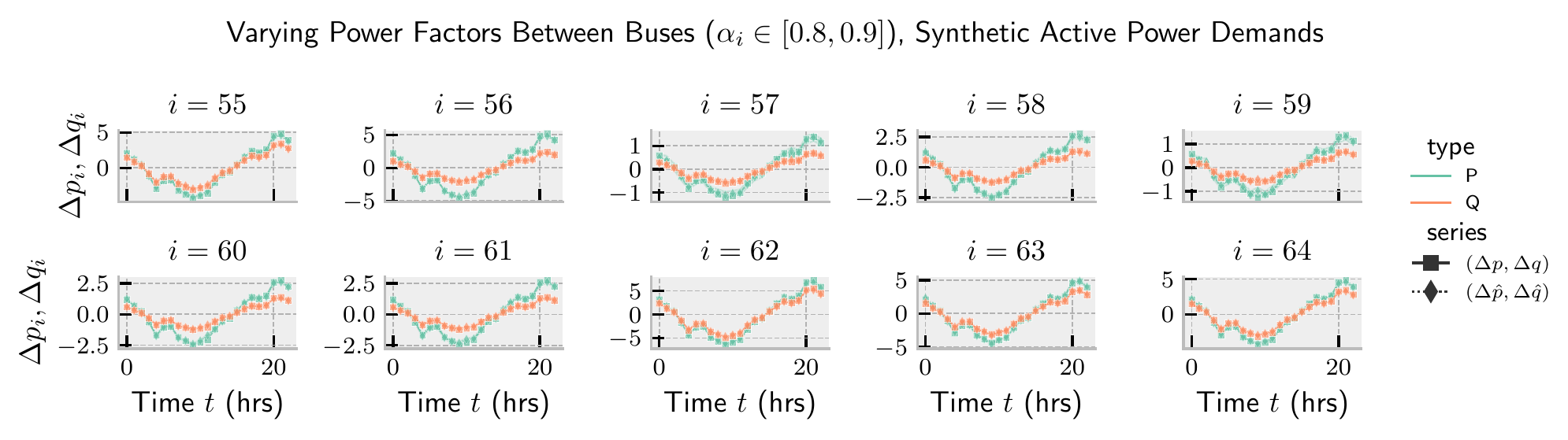}
        \vspace{-3mm}
        \caption{Estimating complex power using voltage magnitudes for the IEEE 123-bus case with varying bus power factors using  \eqref{eq:basic_estimation_computation}.}
        \label{fig:implicit_complex_estimation_results_ieee123_vary_pf}
    \end{figure*}
    
    Subsequently, we evaluate the $\tilde{\mathbf{S}}$ matrix recovery technique using the IEEE 123-bus test case with multiple reactive power behaviors. The bus power factors are set to 0.9 for all loads. We initialize a random $\tilde{\mathbf{S}}_0$ with 90\% to 25\% of the entries unknown. Applying the matrix recovery algorithm with hyperparameters $\lambda = 0.125$ and $\delta = 0.06$, we estimate a wide $\tilde{\mathbf{S}}^{\#}$ with a relative percentage error, 
    $
        (\| \tilde{\mathbf{S}} - \tilde{\mathbf{S}}^{\#} \|
    \big/ \| \tilde{\mathbf{S}} \|) \times 100$ of 7.62\% when 20\% of the entries are known. The recovered matrix is illustrated in Fig. \ref{fig:my_label}. The performance of the estimated matrix as the system evolves across time is shown in Fig. \ref{fig:ieee123_rel_err_series}. Note that this does not depict a time-varying estimate of the matrix, in contrast with Fig. \ref{online_results}.
    \begin{figure}
        \centering
        \includegraphics[width=0.95\linewidth, keepaspectratio]{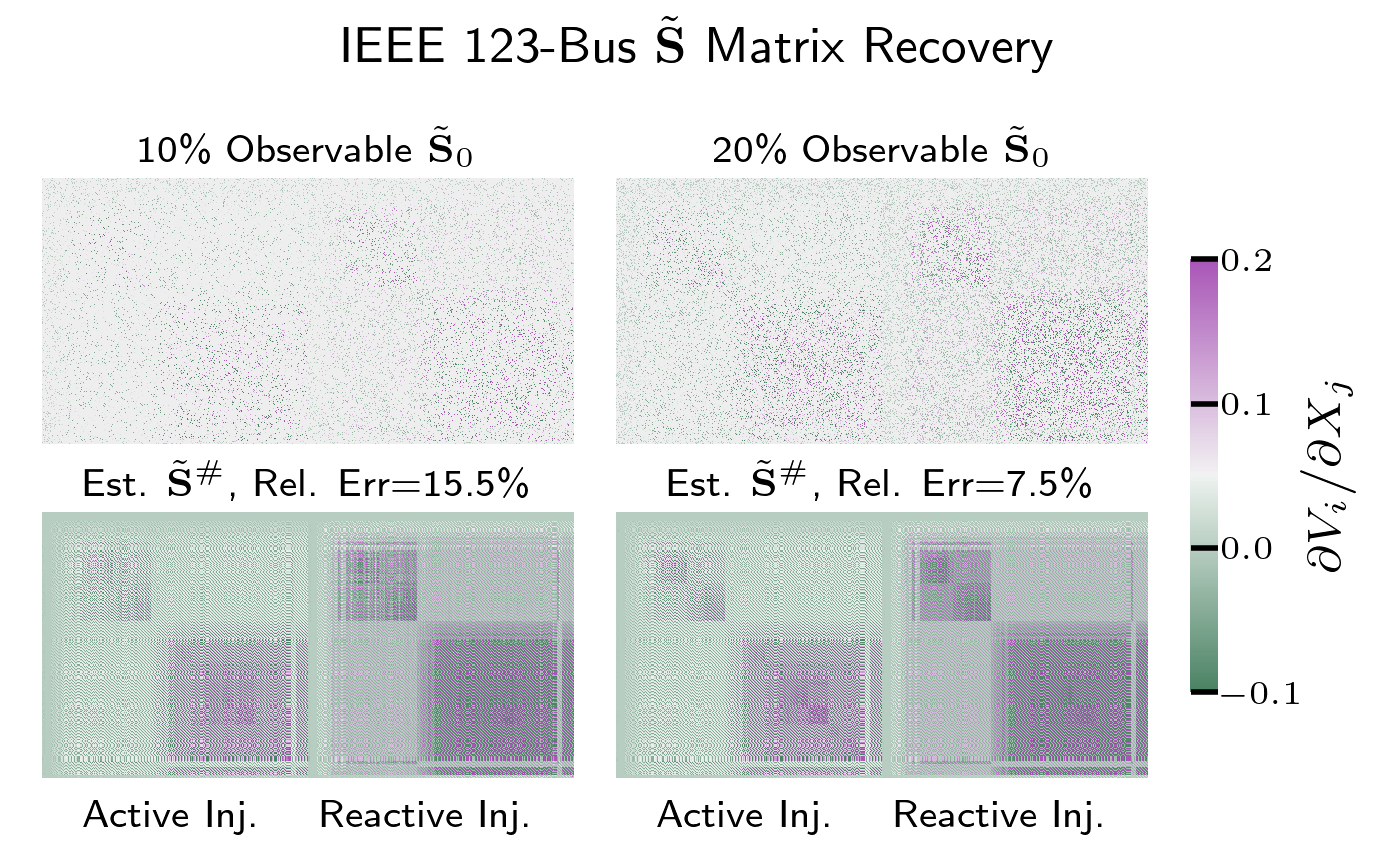}
    \caption{Recovering the $\tilde{\mathbf{S}}$ matrix for the IEEE 123-bus test feeder using \eqref{eq:partial_information_nuc_reg_lsq} with 80\% and 90\% of the 150,152 coefficients unknown. Hyperparameters are $\lambda = 0.125$ and $\delta = 0.06$. Rel. Fro. error $(\| \tilde{\mathbf{S}} - \tilde{\mathbf{S}}^{\#} \|
    \big/ \| \tilde{\mathbf{S}} \|) \times 100 = 7.62\%$.}
        \label{fig:my_label}
    \end{figure}

    \begin{figure}[t]
     \centerline{\includegraphics[width=0.95\linewidth,height=.3\textheight,keepaspectratio]{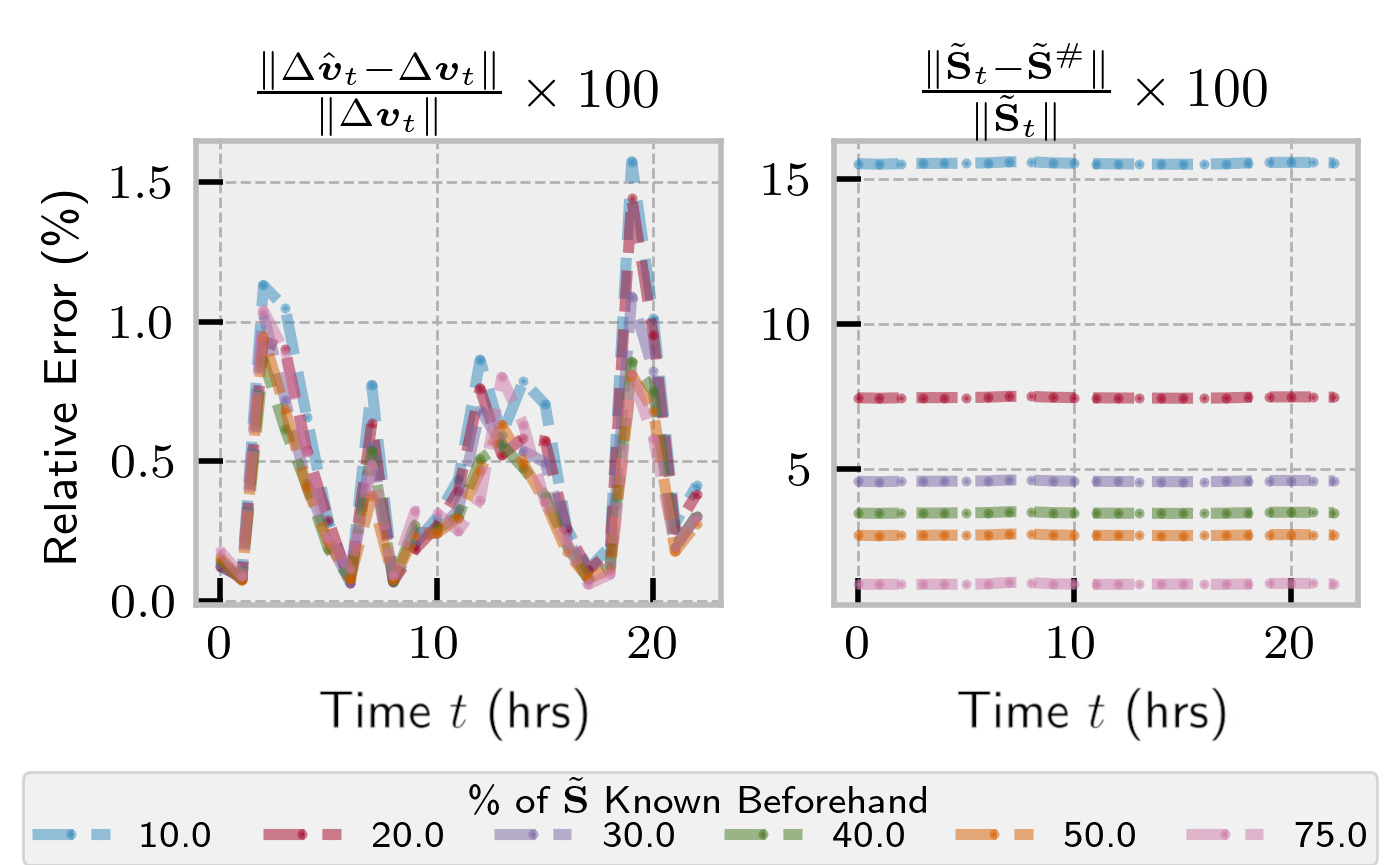}}
        \caption{Performance of the recovery of the $\tilde{\mathbf{S}}$ matrix for the IEEE 123-bus test feeder with 90\%-25\% of the 150,152 coefficients unknown, using \eqref{eq:partial_information_nuc_reg_lsq} with hyperparameters $\lambda = 0.125$ and $\delta = 0.06$.}
        \label{fig:ieee123_rel_err_series}
    \end{figure}

    \subsection{Tests on large-scale meshed networks}
    \label{sec:large_scale_experiment}
        This section details the results of experimentally verifying the proposed conditions on large-scale, more realistic, meshed test networks. The proposed conditions appear to hold for larger test networks at the AC power flow solution.

        Specifically, experiments mirroring those done in Section \ref{sec:comp_test_power_factor_bound} and Section \ref{sec:sens_matrix_recovery} are conducted on two larger scale test networks in an effort to more realistically capture the behavior of true electric power networks.
        
        First, we use the 73-bus IEEE Reliability Test System-Grid Modernization Lab Consortium (RTS-GMLC) network model \cite{barrows_rts_gmlc},  a publicly available and open source.
        This network model seeks to be a standard and open test case for electric power system production cost modeling and reliability calculations. 
        
        Second, we use the 2k-Bus Synthetic Texas Model from the Texas A\&M University/ARPA-E PERFORM  \cite{birchfield-synthetic-grids} synthetic grid dataset. We elect to study the 2k-Bus Synthetic Texas Model for the purpose of ensuring that the experiments can be efficiently reproduced on consumer hardware. In our experiments, we use a ThinkPad T14 laptop computer with a Ryzen 7 PRO 4750-U 8-core 1.7 GHz CPU with 42 GB of RAM. 
        The results of these experiments are shown in Table \ref{table:invertibility_bound_meshed_large_scale}.
    \begin{table*}[ht]
      \renewcommand{\arraystretch}{1.45}
        \vspace{-3mm}
        \tabcolsep=0.11cm
       \centering
        \caption{Numerical validation of Theorems for large-scale meshed \textsc{Matpower} test cases at default operating point}
        \begin{tabular}{|c||c||c||c||c||c||c||c||c|c|}
          \hline
          & \multicolumn{2}{c||}{\textbf{Assumption Holds?}} &  \multicolumn{5}{c||}{\textbf{Quantity}}  &  \multicolumn{2}{c|}{\multirow{1}{*}{\textbf{Thm. Holds?}}}\\
        \cline{1-10}
        \textbf{Case } 
        &  $\dpdth \succ 0$ 
        & $\mathbf{J}^{-1}$ Exists 
        & $\lambda_{\rm min}(\mathbf{J})$
        & $\alpha_{\text{max}}- \alpha_{\text{min}}$ 
        & $k_{\text{max}}-k_{\text{min}}$
        & $\norm{\mathbf{M}^{-1}}^{-1}_2 \big\lVert \dpdth \big\rVert_2^{-1}$ 
        & $\|\mathbf{M}^{-1} \Delta \mathbf{K} \dpdth \|_2$ 
         & Thm. \ref{thm1:suff_cond} 
        & Thm. \ref{thm:eigenvalue_checks}\\
        \hline
             RTS-GMLC & Yes & Yes&    0.295  &   0.7632 & 0.2923 & 0.0324 & 0.02289 & \textcolor{ForestGreen}{Yes} & \textcolor{ForestGreen}{Yes}\\
        \hline
        Sg2k & Yes & Yes&    0.235  &   0.192 & 0.591 & $0.207 \times 10^{-3}$ & 0.1355 & \textcolor{ForestGreen}{Yes} & \textcolor{ForestGreen}{Yes}\\
        \hline
        \end{tabular}
    
        \label{table:invertibility_bound_meshed_large_scale}
    \end{table*}

\section{Limitations and Future Work}
The proposed theory and algorithms, as well as the presented case studies have limitations and opportunity for future work, which we describe throughout this section. 

\subsection{Analytical Results}
The analytical results developed in this paper rely on the definition of the sensitivities of voltage phasors to complex power injections in rectangular coordinates \eqref{eq:vph_sensitivities_definition} developed in \cite{christakou_efficient_2013}. According to \cite{christakou_efficient_2013}, these definitions are valid for radial electric power systems, and we have mirrored this scope in Lemmas \ref{lemma:representation} and \ref{lemma:distinct_vmag_sens}. There is an opportunity for future work to extend those analyses and generalize to non-radial networks. 

Furthermore, we stress that the proposed condition \eqref{eq:inv_condition_2} for the existence of a unique $\Delta \boldsymbol{x}$ is a sufficient but not necessary condition. While this sufficient condition does hold at the default operating point for numerous test cases as shown in Tables \ref{table:invertibility_bound_radial} and \ref{table:invertibility_bound_meshed}, the fact that it does not hold for other test cases does not necessarily imply that the applications are impossible. This does, however, indicate that this problem is not fully solved from a theoretical perspective. Reformulations and tightening of this inequality or the development of necessary conditions is an important direction for future work.

There is also an opportunity for future work to explicitly extend the analysis outlined here to a broader range of network topologies. Notably, we have completely neglected the explicit multi-phase analysis of the sensitivities. We hypothesize that the models and theories we have considered are general for multi-phase distribution networks; however, their explicit formulation remains an opportunity for future work. 

In some topologies, additional considerations for details beyond those in the scope of this paper may be needed. On the other hand, some topologies such as secondary distribution networks\textemdash i.e., networks on the secondary side of a distribution network service transformer\textemdash may contain computationally favorable topological structures. In general, prior knowledge of the network topology may admit a host of new extensions to this research. 

Finally, we note that a valuable direction for future work is to extend the conditions derived in this work to evaluate state estimation feasibility if a subset of the buses have \emph{only} voltage magnitude measurements or \emph{only} power measurements. Related work on topology estimation with partially observable AMI/smart meter measurements, e.g., \cite{lin_data-driven_2021}, may indicate that this open question is promising.

\subsection{Computational Results}
In practice, the quantity $\Delta k$ on the left hand side of the inequality derived in Theorem \ref{thm1:suff_cond}, \eqref{eq:inv_condition_2}, is inherently time-varying, as it depends on the operating point of the network. The right hand side of the inequality is also time-varying, as it depends on the angle sensitivity matrices from the power flow Jacobian, which are themselves functions of the operating point. In the results presented in Table~\ref{table:invertibility_bound_radial} and Table~\ref{table:invertibility_bound_meshed}, we use the default power injections of the test cases; therefore, the results represent a single point in time. There is an opportunity for future work to investigate how these quantities change across time and how this impacts the estimation quality.

The sensitivity matrix completion methods have clear practical applications\textemdash particularly in realistic AMI data modeling scenarios where missing data are prevalent\textemdash however, the nuclear norm regularizer contained in the objective function of these problems, $\lambda \norm{\mathbf{S}}_*$, can be computationally expensive for large circuits when na\"ive implementations are used. Future improvements in the computational efficiency of nuclear norm-regularized optimization problems have applications in some of the results of this work. This could improve the feasibility of these results to large-scale distribution system problems, and thus, a wider range of real-world settings.

\section{Discussion}
\label{sec:discussion}
The theory and algorithms developed in this paper directly relate to one another. The conditions provided by Theorems \ref{thm1:suff_cond} and \ref{thm:eigenvalue_checks} ensure when it is possible to relate $\vv$ to $\vp$ and $\vq$ without knowledge of $\vtheta$ through the underdetermined system \eqref{eq:encoded-matrices}. Correspondingly, the power system applications of the matrix completion/recovery algorithms reviewed in Section \ref{sec:algorithms_application} allows engineers to go further, and indicate that is possible to infer of the voltage sensitivities of nodes that do not have measurements of $\vp,\vq$, or even $\vv$, provided that the following assumptions hold:
\begin{enumerate}
    \item The engineer has access to ``seed" coefficients that form $\tilde{\mathbf{S}}_0$. For example, this could take the form of precomputed voltage sensitivity coefficient estimates computed from historical data of a limited number of measured nodes elsewhere in the distribution network.
    \item Some prior knowledge, intuition, or estimate of the dimensionality of the of the full $\mtx{\Tilde{S}}$ matrix to be estimated. This could take the form of topological information reported from engineers or technicians, GPS data, and/or satellite imagery.
    \item The mild technical assumption described in  Remark \ref{remark:approx_low_rank}.
\end{enumerate}

The results discussed in Section \ref{sec:analytical_results} directly inform the results of Section \ref{sec:algorithms_application} by allowing Algorithms such as those developed in Section \ref{sec:algorithms_application} may be of valuable practical application to solving engineering problems in rural distribution system environments, due to lagging availability of the very fast GPS synchronization and robust communication infrastructure required for some classes of control and optimization algorithms. Thus, we conclude that this research provides valuable engineering insight for enhancing the observability of distribution grids in austere environments.


    
        
        
        
        

\section{Conclusion}


This paper developed theory and algorithms for analyzing sensitivity matrices relating voltage magnitudes to active and reactive power injections, and correspondingly, analyzing complex power injections with voltage magnitudes. 

First, we showed that these matrices achieve distinct values in radial distribution networks. Then, we developed a sufficient condition based on the bus power factors for arbitrary networks that guarantees a solution to the underdetermined linear system formed by these matrices. In summary, the condition shows that there exists a \textit{sufficient margin of distance} between the bus power factors within the network that, if satisfied, allows for the estimation of unique active and reactive power injection vectors that explains the vector of voltage magnitude perturbations, without observation of the voltage phase angles. 

Finally, several algorithms were proposed for estimating voltage sensitivity matrices or updating an existing estimate from measurement data in diverse network models. Such algorithms may help enable the application of voltage sensitivities to decision-making problems in real-world electric power distribution systems.

The results of this paper indicate that some distribution networks may be able to be modeled with significantly reduced data input requirements. Problems that seemingly require phase angle measurements may be able to be solved using voltage magnitude measurements by exploiting the voltage sensitivities created by active and reactive power injections. Further, existing linear sensitivity models for the bus voltages can be updated and used to provide information about other buses in the network, assisting the identification and model-free prediction of the behavior of distribution grids in settings where measurements are limited in both number and quality.

\begin{appendices}

    \section{Proofs of Lemmas}
    
        \subsection{Proof of Lemma \ref{lemma:representation}}
        \label{apdx:proof_of_vmag_representation}
        \begin{proof}
            Let $v_i$ and $\theta_i$ be the real-valued magnitude and angle of the voltage phasor. Write the voltage phasor sensitivity at bus $i$ to an arbitrary quantity $x$ as:
            \begin{equation}
                \frac{\partial \bar{v}_i }{\partial x } = \frac{\partial}{\partial x }\{ v_i e^{j \theta _i}\} = \frac{\partial v_i}{\partial x } e^{j \theta_i} + j v_i \frac{\partial \theta_i}{\partial x }e^{j \theta_i},
            \end{equation}
            so, we have that:
            {\small
            \begin{align}
            e^{-j \theta_i} \frac{\partial \bar{v}_i }{\partial x }  = \frac{\partial v_i}{\partial x } + j v_i \frac{\partial \theta _i}{\partial x }
            \iff \operatorname{Re}\Bigg\{  
            e^{-j \theta_i} \frac{\partial \bar{v}_i}{\partial x }
            \Bigg\} &= \frac{\partial v_i}{\partial x }.
            \end{align}
            }
            Next, observe that $\bar{v}_i^*/{v_i} = e^{-j \theta}$. Therefore,
            \begin{align}
                \operatorname{Re}\Bigg\{  
            e^{-j \theta_i} \frac{\partial \bar{v}_i}{\partial x }
            \Bigg\} &= \frac{1}{v_i} \operatorname{Re} \Bigg\{ 
                \bar{v}_i^* \frac{\partial\bar{v}_i}{\partial x }
            \Bigg\}, \\
                \frac{\partial v_i}{\partial x } &= \frac{1}{v_i} \operatorname{Re} \Bigg\{ 
                \bar{v}_i^* \frac{\partial\bar{v}_i}{\partial x }
            \Bigg\}.
            \end{align}
            Make $x = p_l$ or $x = q_l$ to get the desired result.
            \end{proof}
        
        \subsection{Proof of Lemma \ref{lemma:distinct_vmag_sens}}
        \label{apdx:proof_of_distinct_vmag_sens}
            \begin{proof}
                Using \eqref{magnitude_sens}, the voltage magnitude sensitivity coefficients for bus $i$ to reactive power at bus $l$ is:
                \begin{equation}
                \label{step1_real}
                     \frac{\partial v_i}{\partial q_l} = \frac{1}{v_i} \operatorname{Re} \left\{{(\operatorname{Re}\{\bar{v}_i\}-j\operatorname{Im}\{\bar{v}_i\})\frac{\partial\bar{v}_i}{\partial q_l}}\right\}, 
                \end{equation}
                and in the same way, for active power, we have:
                \begin{equation}
                \label{step1_reactive}
                    \frac{\partial v_i}{\partial p_l} = \frac{1}{v_i} \operatorname{Re} \left\{(\operatorname{Re}\{\bar{v}_i\}-j\operatorname{Im}\{\bar{v}_i\})\frac{\partial\bar{v}_i}{\partial p_l}\right\}. 
                \end{equation}
                Simplifying the above results in:
                \begin{align}
                    \frac{\partial v_i}{\partial q_l} &= \frac{1}{v_i}(\operatorname{Re}\{\bar{v}_i\}a+\operatorname{Im}\{\bar{v}_i\}b),\\
                    \frac{\partial v_i}{\partial p_l} &=\frac{1}{v_i}(\operatorname{Re}\{\bar{v}_i\}c+\operatorname{Im}\{\bar{v}_i\}d).
                \end{align}
                From Remark \ref{remark:unique_complex_vph_sens}, if $\frac{\partial\bar{v}_i}{\partial q_l} \neq \frac{\partial\bar{v}_i}{\partial p_l}$ for nonzero solutions this implies that either $a \neq c$ or $b \neq d$. So, provided that the sensitivities are not zero, i.e., \eqref{zero_con_1} and \eqref{zero_con_2} are satisfied, then it must also be true that $ \frac{\partial v_i}{\partial q_l} \neq \frac{\partial v_i}{\partial p_l} \ \forall i, l$.  
                \end{proof}
        
        \subsection{Proof of Lemma \ref{lemma:k_inverse_func}}
        \label{apdx:proof_of_k_inv}
        \begin{proof}
            
            Recall that we defined  $k : \alpha \in (0,1] \mapsto \mathbb{R}$ as $k(\alpha) \triangleq \pm \frac{1}{\alpha} \sqrt{1- \alpha^2}$, where $q(p|\alpha) = k(\alpha)p$. We want to show that there exists a $k^{-1}(\cdot): k(\alpha) \mapsto \alpha$, where $k^{-1}(k(\alpha)) = \alpha$ for any $\alpha \in (0,1]$. Note that $k^2(\alpha) = \frac{1- \alpha^2}{\alpha^2}$. Thus we have that
            \begin{subequations}
                \label{eq:final_steps}
                \begin{align}
                    \alpha^2k^2(\alpha) = 1-\alpha^2,\\
                    \iff \left(k^2(\alpha) + 1\right) \alpha^2 -1 = 0,\\
                    \label{eq:final-step-kinv-proof}
                    \iff \alpha = \pm \sqrt{\frac{1}{k^2(\alpha)+1}}.
                \end{align}
            \end{subequations}
            By definition, the power factor must satisfy $\alpha \in (0,1]$, thus, we can discard the negative solution in \eqref{eq:final-step-kinv-proof}. Finally, we have
            \begin{equation}
                k^{-1}(\alpha) = \sqrt{\frac{1}{k^2(\alpha) +1}} \quad \alpha \in (0,1],
            \end{equation}
            which is what we wanted to show.
        \end{proof}

   \section{Additional numerical results for Theorem 2}
   \label{apdx:additional-numerical-results}
    \begin{table}[h!]
        \renewcommand{\arraystretch}{1.25}
                \caption{Additional numerical results for Theorem \ref{thm:eigenvalue_checks} for radial \textsc{Matpower} test cases at default operating point}
                \label{table:thm2_radial_test}
                \vspace{-3mm}
                \tabcolsep=0.07cm
                
            \begin{center}
                \begin{tabular}{|c|l|l|l|l|l|l|}
                \hline
                \multicolumn{1}{|c|}{\multirow{2}{*}{\textbf{Case}}} 
                    & \multicolumn{2}{l|}{$\mathbf{K}$ Matrix}
                    & \multicolumn{2}{l|}{Minimum eigenvalue} & \multicolumn{2}{l|}{Power Factors}
                    \\ 
                    \cline{2-7} 
                    \multicolumn{1}{|c|}{}                  
                    & \multicolumn{1}{l|}{$k_{\rm min}$} 
                    & \multicolumn{1}{l|}{$k_{\rm max}$} 
                    & \multicolumn{1}{l|}{$\stildedag$} & \multicolumn{1}{l|}{$ \stildestar$} 
                    & \multicolumn{1}{l|}{$\alpha_{\rm min}$} 
                    & \multicolumn{1}{l|}{$\alpha_{\rm max}$} 
                    \\ 
                    \hline
                    2 &

                 0.285 & 
                 0.285 & 
                 0.022 & 
                 0.076 &
                 0.932 &
                 0.962 
                 \\
                 \hline
                    4\_dist & 
                 0.5 & 
                 0.5 & 
                 $2.3 \times 10^{-3}$ & 
                 $4.6 \times 10^{-3}$ & 
                 0.894 & 
                 0.894  
                \\
                \hline
                      5\_tnep & 
                 0.329 & 
                 0.329 & 
                 $4.57 \times 10^{-3}$ & 
                 $13.9 \times 10^{-3}$ & 
                 0.949 & 
                 0.949  
                 \\ 
                \hline
                 12da & 
                 0.75 & 
                 1.0 & 
                 $3.87 \times 10^{-3}$ & 
                 $4.59 \times 10^{-3}$ & 
                 0.707 & 
                 0.80  
                 \\ 
                \hline
                   15da & 
                 1.020 & 
                 1.020 & 
                 $3.59 \times 10^{-3}$ & 
                 $3.52 \times 10^{-3}$ & 
                 0.7 & 
                 0.7  
                 \\ 
                \hline
                  16am & 
                 0.133 & 
                 0.90 & 
                 $4.265 \times 10^{-10}$ & 
                 $6.239 \times 10^{-10}$ & 
                 0.673 & 
                 1.0  
                 \\ 
                \hline
                 22 & 
                 0.860 & 
                 1.355 & 
                 $0.119 \times 10^{-3}$ & 
                 $0.137 \times 10^{-3}$ & 
                 0.594 & 
                 0.758  
                 \\ 
                \hline
                  28da & 
                 1.020 & 
                 1.020 & 
                 $1.46 \times 10^{-3}$ & 
                 $1.43 \times 10^{-3}$ & 
                 0.699 & 
                 0.700  
                 \\ 
                \hline
                33mg & 
                 0.167 & 
                 3.0 & 
                 $0.40 \times 10^{-3}$ & 
                 $0.69 \times 10^{-3}$ & 
                 0.316 & 
                 0.986  
                 \\ 
                \hline
                70da & 
                 0.40 & 
                 0.923 & 
                 $1.03 \times 10^{-3}$ & 
                 $1.76 \times 10^{-3}$ & 
                 0.735 & 
                 0.929  
                 \\ 
                \hline
                69 & 
                 0.583 & 
                 0.846 & 
                 $2.98 \times 10^{-5}$ & 
                 $4.21 \times 10^{-5}$ & 
                 0.763 & 
                 1.0  
                 \\ 
                \hline
                85 & 
                 1.02 & 
                 1.02 & 
                 $0.424 \times 10^{-3}$ & 
                 $0.416 \times 10^{-3}$ & 
                 0.69 & 
                 1.0  
                 \\ 
                \hline
                16ci & 
                 0.36 & 
                 0.90 & 
                 $1.99 \times 10^{-3}$ & 
                 $3.97 \times 10^{-3}$ & 
                 0.743 & 
                 0.940  
                 \\ 
                \hline
                74ds & 
                 0.571 & 
                 1.0 & 
                 $9.353 \times 10^{-5}$ & 
                 $0.1398 \times 10^{-3}$ & 
                 0.707 & 
                 0.863  
                 \\ 
                \hline
                51ga & 
                 0.299 & 
                 1.0 & 
                 $0.679 \times 10^{-3}$ & 
                 $1.115 \times 10^{-3}$ & 
                 0.707 & 
                 0.958 
                 \\ 
                \hline
                  15nbr & 
                 1.0202 & 
                 1.0202 & 
                 $0.246$ & 
                 $0.241$ & 
                 0.70 & 
                 0.70  
                 \\ 
                \hline
                 17me & 
                 0.133 & 
                 0.90 & 
                 $24.12 \times 10^{-3}$ & 
                 $39.66 \times 10^{-3}$ & 
                 0.673 & 
                 0.991  
                 \\ 
                \hline
                  18nbr & 
                 1.020 & 
                 1.0204 & 
                 $0.245$ & 
                 $0.2399$ & 
                 0.699 & 
                 0.700  
                 \\ 
                \hline
                 34sa & 
                 1.563 & 
                 1.80 & 
                 $0.236 \times 10^{-3}$ & 
                 $0.142 \times 10^{-3}$ & 
                 0.486 & 
                 1.0  
                 \\ 
                \hline
                        85 & 
                 1.02 & 
                 1.02 & 
                 $0.424 \times 10^{-3}$& 
                 $0.416 \times 10^{-3}$& 
                 0.699 & 
                 1.0  
                \\
                \hline
                   94pi & 
                 0.5 & 
                 0.5 &
                 $0.20 \times 10^{-3}$ & 
                 $0.41 \times 10^{-3}$ & 
                 0.898 & 
                 1.0  
                 \\ 
                \hline
                136ma & 
                 0.398 & 
                 0.489 & 
                 $9.71 \times 10^{-5}$ & 
                 $0.23 \times 10^{-3}$ & 
                 0.898 & 
                 1.0  
                 \\ 
                \hline
                  141 & 
                 0.6197 & 
                 0.6197 & 
                 $2.147 \times 10^{-7}$ & 
                 $3.464 \times 10^{-7}$ & 
                 0.85 & 
                 1.0  
                 \\ 
                \hline
                \multicolumn{7}{l}{} 
                \end{tabular}
            \end{center}
        \end{table}

\end{appendices}

\section*{Acknowledgement}
The authors thank the reviewers for their valuable feedback that significantly improved the manuscript. S. Talkington and S. Grijalva thank M. J. Reno, J. A. Azzolini, and D. Pinney for helpful commentary and discussions.
\balance

\bibliographystyle{IEEEtran}
\bibliography{refs/refs.bib,refs/extras.bib,refs/r2_refs.bib}

\vskip -2.5\baselineskip plus -1fil
\begin{IEEEbiographynophoto}{Samuel Talkington} received the B.S. degree in Electrical Engineering from West Virginia University, Morgantown, WV, USA in 2020. He is currently a Ph.D. student at the School of Electrical and Computer Engineering, Georgia Institute of Technology, Atlanta, GA, USA where he is a National Science Foundation Graduate Research Fellow. In 2019 and 2020, he completed interships with the Power System Studies and Electrical Distribution Divisions,  Mitsubishi Electric Power Products, Inc., Warrendale, PA, USA. In 2022 he was a research intern with the Critical Infrastructure Resilience program, Lawrence Livermore National Laboratory, Livermore, CA, USA. His research interests are decision-making and control in electric power systems.
\end{IEEEbiographynophoto}

\vskip -2.5\baselineskip plus -1fil
\begin{IEEEbiographynophoto}{Daniel Turizo} received the B.S and M.S.c. degrees in Electrical Engineering from the Universidad del Norte, Barranquilla, Colombia in 2016 and 2018, respectively. He is currently a Ph.D. student at the School of Electrical and Computer Engineering, Georgia Institute of Technology, Atlanta, GA, USA where he is a Fullbright Fellow. From 2018-2020 he was an Adjunct Professor of Electrical Engineering at the Universidad del Norte, Barranquilla, Colombia. 
\end{IEEEbiographynophoto}

\vskip -2.5\baselineskip plus -1fil
\begin{IEEEbiographynophoto}{Santiago Grijalva} is the Southern Company Distinguished Professor of Electrical and Computer Engineering and Director of the Advanced Computational Electricity Systems (ACES) Laboratory at the Georgia Institute of Technology. His research interest is on decentralized power system control, cyber-physical security, and economics. From 2002 to 2009 he was with PowerWorld Corporation. From 2013 to 2014 he was with the National Renewable Energy Laboratory (NREL) as founding Director of the Power System Engineering Center (PSEC). Dr. Grijalva was a Member of the NIST Federal Smart Grid Advisory Committee. His graduate degrees in ECE are from the University of Illinois at Urbana-Champaign. 
\end{IEEEbiographynophoto}

\vskip -2.5\baselineskip plus -1fil
\begin{IEEEbiographynophoto}{Jorge Fernandez} received the B.S. and M.S.c. degrees in Electrical Engineering from the Universidad Tecnológica de Pereira, Pereira, Colombia in 2017 and 2019, respectively. He is currently a Ph.D. student at the School of Electrical and Computer Engineering, Georgia Institute of Technology, Atlanta, GA, USA where he is a Fullbright Fellow. From 2020 to 2022 he was a digital services engineer with a Colombian ISO.
\end{IEEEbiographynophoto}

\vskip -2.5\baselineskip plus -1fil
\begin{IEEEbiographynophoto}{Daniel K. Molzahn} is an Assistant Professor in the School of Electrical
and Computer Engineering at the Georgia Institute of Technology and also holds an appointment as a computational engineer in the Energy Systems Division at Argonne National Laboratory. He was a Dow Postdoctoral Fellow in Sustainability at the University of Michigan, Ann Arbor. He received the B.S., M.S., and Ph.D. degrees in electrical engineering and the Masters of Public Affairs degree from the University of Wisconsin–Madison, where he was a
National Science Foundation Graduate Research Fellow. He received the IEEE Power and Energy Society’s Outstanding Young Engineer Award in 2021 and the NSF CAREER Award in 2022.
\end{IEEEbiographynophoto}


\end{document}